%% file: main-mining.tex
\documentclass{sig-alternate-10pt-altered}
\usepackage[svgnames]{xcolor}
\usepackage{colonequals}
\usepackage{paralist}
\usepackage{multirow}
\usepackage{booktabs}
\usepackage{tabularx}
\usepackage{cellspace}
\usepackage{xfrac}
\usepackage{caption}
\usepackage{subcaption}
\usepackage{needspace}
\usepackage[vlined,boxed,noend]{algorithm2e}

\SetAlCapSkip{0.5em}
\usepackage[auth-lg]{authblk}
\let\chapter\undefined

\newtheorem{theorem}{Theorem}[section]
\newtheorem{proposition}[theorem]{Proposition}
\newtheorem{lemma}[theorem]{Lemma}
\newtheorem{corollary}[theorem]{Corollary}
\newtheorem{example}[theorem]{Example}
\newtheorem{definition}[theorem]{Definition}

\newcommand{\var}[1]{\operatorname{\mathit{#1}}}
\newcommand{\name}[1]{\operatorname{#1}}
\newcommand{\funct}[2]{\operatorname{#1}\!\left({#2}\right)}
\newcommand{\functsq}[2]{\operatorname{#1}\!\left[{#2}\right]}

\newcommand{\card}[1]{\left|{#1}\right|}

\newcommand{\itm}[1]{\textsf{#1}}
\newcommand{\txtitm}[1]{#1}
\renewcommand{\epsilon}{\varepsilon}
\renewcommand{\phi}{\varphi}

\newcommand{\oof}[1]{\operatorname{O}\!\left({#1}\right)}
\newcommand{\oofi}[1]{\operatorname{O}({#1})}
\newcommand{\omegaof}[1]{\operatorname{\Omega}\!\left({#1}\right)}
\newcommand{\omegaofi}[1]{\operatorname{\Omega}({#1})}
\newcommand{\thetaof}[1]{\operatorname{\Theta}\!\left({#1}\right)}

\newcommand{\ideal}[1]{\,\downarrow\!\!{#1}}
\newcommand{\fpsp}{\mathrm{FP}^{\mathrm{\#P}}}
\newcommand{\true}{\textsf{true}}
\newcommand{\false}{\textsf{false}}

\newcommand{\supp}[2]{\name{supp}_{\!_{#1}}\!\left({#2}\right)}
\newcommand{\freq}[1][\cdot]{\name{freq}}
\newcommand{\freqf}[1][\cdot]{\funct{\freq}{#1}}
\newcommand{\wid}[1][\ont]{\functsq{w}{#1}}
\newcommand{\mfi}[1][\freq]{\var{mfi}}
\newcommand{\mii}[1][\freq]{\var{mii}}

\newcommand{\idom}{\mathcal{I}}
\newcommand{\prop}[1]{\mathcal{AC}\!\left[#1\right]}

\newcommand{\ont}{\Psi}
\newcommand{\ifunct}[1]{\funct{I}{#1}}
\newcommand{\onti}{\ifunct{\ont}}
\newcommand{\ikfunct}[2]{\funct{I^{\left(#1\right)}}{#2}}
\newcommand{\ontik}[1][k]{\ikfunct{#1}{\ont}}
\newcommand{\onts}{\funct{S}{\ont}}
\newcommand{\ontsk}[1][k]{\funct{S^{\left(#1\right)}}{\ont}}
\newcommand{\idomok}[1][k]{\mathcal{AC}^{(k)}\!\left[\ont\right]}
\newcommand{\idomo}{\mathcal{I}}
\newcommand{\idomop}{\prop{\ont}}
\newcommand{\relt}[1]{\leq}
\newcommand{\rel}[1]{\lessdot}
\newcommand{\orelt}{\relt{\ont{}}}
\newcommand{\orel}{\rel{\ont{}}}
\newcommand{\orelit}{\relt{\onti{}}}
\newcommand{\oreli}{\rel{\onti{}}}
\newcommand{\orelikt}[1][k]{\relt{\onti{}\!,#1}}
\newcommand{\orelik}[1][k]{\rel{\onti{}\!,#1}}

\newcommand{\concat}{\circ}

\newcommand{\aprob}{\textup{\texttt{MineFreq}}}

\title{On the Complexity of Mining Itemsets from the Crowd Using Taxonomies\titlenote{This work has been partially funded by the European Research Council under the FP7, ERC grant MoDaS, agreement 291071, and by the Israel Ministry of Science.}}
\author[1,2]{Antoine Amarilli}
\author[1]{Yael Amsterdamer}
\author[1]{Tova Milo}
\affil[1]{{\normalsize Tel Aviv University, Tel Aviv, Israel}}
\affil[2]{{\normalsize \'{E}cole normale sup\'{e}rieure, Paris, France}}

\usepackage[pdfborder={0 0 0}, plainpages, pdfpagelabels=false, pdfstartview=FitH]{hyperref}

\begin{document}

\clubpenalty=10000
\widowpenalty = 10000

\maketitle

\input{intro}
\input{prelim}

\input{query-complexity}

\input{computational-complexity}
\input{heuristics}

\input{greedy}
\input{related}
\input{conc}

\bibliographystyle{abbrv}
{\small
\bibliography{main-mining}}
\clearpage
\input{appendix}

\end{document}

%% file: intro.tex
\begin{abstract}
We study the problem of frequent itemset mining in domains where
data is not recorded in a conventional database but only exists in human knowledge.
We provide examples of such scenarios, and present a crowdsourcing model for them. The model uses the crowd as an
oracle to find out whether an itemset is frequent or not,
and relies on a known taxonomy of the item domain to guide the search for frequent itemsets.
In the spirit of data mining with oracles, we analyze the complexity of this problem in terms of
(i) crowd complexity, that measures
 the number of crowd questions required to identify the
frequent itemsets; and (ii) computational complexity, that measures the
computational effort required
to choose the questions. We provide lower and upper complexity bounds in terms of the
size and structure of the input taxonomy, as well as the size of
a concise description of the output itemsets.
We also provide constructive algorithms that achieve the upper
bounds, and consider more efficient variants
for practical
situations.
\end{abstract}

\section{Introduction}
\label{sec:intro}
The identification of \emph{frequent itemsets}, namely
sets of items that frequently occur together, is a basic ingredient
in data mining algorithms and is used
to discover
interesting patterns in large data sets~\cite{agarwal1994fast}. A common assumption in such
algorithms is that the transactions
to be mined (the sets of co-occurring items)
have been recorded and are stored in a database. In
contrast, there is data which is not recorded in a systematic manner, but
only exists in human knowledge. Mining this type of data is the goal of
this paper.

As a simple example, consider a social scientist analyzing the life
habits of people, in terms of activities (watching TV, jogging,
reading, etc.)\ and their contexts (time, location, weather, etc.).
Typically, for large communities, there is no comprehensive database that records
all transactions where an individual performs some combination of activities in a certain context. Yet, some trace of the data remains in the memories of
the individuals involved. As another example, consider a health researcher who wants to identify new drugs by analyzing the practices of folk medicine (also known as traditional medicine, i.e., medicinal practice that
is neither documented in writing nor tested out under a scientific protocol):
the researcher may want to determine, for instance, which treatments are often applied together
for a given combination of symptoms. For this purpose too, the main source of knowledge are the
folk healers and patients themselves.

In a previous work
\cite{amsterdamer2013crowd,amsterdamer2013crowdminer} we have
proposed to address this challenge using \emph{crowdsourcing} to
mine the relevant information from the crowd. Crowdsourcing
platforms (such as,
e.g.,~\cite{amsterdamer2013crowdminer,franklin2011crowddb,marcus2011crowdsourced,mturk,parameswaran2012deco})
are an effective tool for harnessing a crowd of Web users to perform
various tasks.
In~\cite{amsterdamer2013crowd,amsterdamer2013crowdminer} we
incorporated crowdsourcing into a \emph{crowd mining framework} for
identifying frequent data patterns in human knowledge,
and demonstrated its efficiency experimentally. The goal of the present paper is to develop the theoretical foundations for crowd mining, and, in particular, to formally study the complexity of identifying frequent itemsets using the crowd.

Before presenting our results, let us explain three important principles that guide our solution.

First, in our settings, \emph{no comprehensive database can be built}.
Not only would it be prohibitively expensive to ask all the relevant people to provide all the required information, but
it is also impossible for people to recall all the details of their
individual transactions
such as activity occurrences,
illnesses,
treatments,
etc.~\cite{amsterdamer2013crowd,bradburn1987answering}. Hence, one cannot simply collect the transactions into a
database that could be mined directly.
Instead, studies show that people do remember some summary
information about their transactions \cite{bradburn1987answering},
and thus, as demonstrated
in~\cite{amsterdamer2013crowd,amsterdamer2013crowdminer}, itemset
frequencies can
be learned by asking the crowd directly about them.

Second, as we want to mine the crowd by posing questions about itemset
frequencies, we must define a suitable \emph{cost model} to evaluate mining
algorithms.
In data mining there are two main approaches for measuring algorithm cost.
The first one (see, e.g.,~\cite{agarwal1994fast}) measures running time, including the cost of accessing the database (database scans), which is not suitable for a crowd setting as there is no database that can be accessed in this manner.
The second approach (see, e.g., \cite{mannila1997levelwise}) assumes the existence of an oracle that can be queried for insights about data patterns (frequency of itemsets, in our case); the cost is then measured by the number of oracle calls. Our setting is closer to this second approach: the crowd serves as an oracle, and we count the number of questions posed to the crowd, namely, \emph{crowd complexity}. In addition, we study \emph{computational complexity}, namely, the time required to compute
the itemsets about which we want to ask the crowd.
There is a clear tradeoff between the costs: investing more computational effort to select questions carefully may reduce the crowd complexity, and vice versa.
See Section~\ref{sec:related} for a further comparison of our work with existing approaches in data mining.

Finally, for the human-knowledge domains that we consider, one can
make mining algorithms more efficient by leveraging semantic knowledge captured by \emph{taxonomies}. A taxonomy in our context is
a partial ``is-a'' relationship on the domain items, e.g.,
\txtitm{tennis} is a \txtitm{sport}, \txtitm{sport} is an \txtitm{activity}, etc. Many such taxonomies are available, both domain-specific (e.g., for diseases~\cite{schriml2012disease}) and general-purpose (e.g., Wordnet~\cite{wordnet}). The use of taxonomies in mining is twofold. First, the semantic dependencies between \emph{items} induce a frequency dependency between \emph{itemsets}: e.g., because \txtitm{tennis} is a \txtitm{sport}, the itemset \{\itm{sunglasses}, \itm{sport}\} implicitly appears in all transactions where \{\itm{sunglasses}, \itm{tennis}\} appears. Hence, if the latter itemset is frequent then so is the former. Second,
with taxonomical knowledge we can avoid asking questions about semantically equivalent itemsets, such as \{\itm{sport}, \itm{tennis}\} and \{\itm{tennis}\}. Taxonomies are known to be a useful tool in data mining~\cite{srikant1995mining} and we study their use under our complexity measures.

\begin{table*}
\centering
\renewcommand{\arraystretch}{1.5}
\noindent\begin{tabularx}{\textwidth}{l@{\hskip 0.4cm}l@{\hskip 1cm}p{6cm}@{\hskip 1cm}X}
\toprule
\multicolumn{2}{c}{} & \normalsize\sf\textbf{With respect to the input} & \normalsize\sf\textbf{With respect to the input and output}\\
\midrule
\multirow{2}{*}{\normalsize\sf\textbf{Crowd}} & \sf\textbf{Lower} & $\omegaof{\log \card{\onts}}$ \hfill {\small(Prop.~\ref{prop:lower-query})} &  \(\omegaof{\card{\mfi}+\card{\mii}}\) \hfill {\small(Prop.~\ref{prop:mfi_lower})}  \\ \cmidrule{2-4}
                                              & \sf\textbf{Upper} & $\oof{\log \card{\onts}}$ \hfill {\small(Prop.~\ref{prop:upper-crowd-input})} & \(\oof{\card{\ont{}}\cdot\left(\card{\mfi}+\card{\mii}\right)}\)  \hfill {\small (Thm.~\ref{thm:mfi_upper})} \\
\midrule
\multirow{2}{*}{\normalsize\sf\textbf{Comp.}} & \sf\textbf{Lower} & $\omegaof{\log \card{\onts}}$ \hfill {\small(Cor. of Prop.~\ref{prop:lower-query})}&  EQ-hard  \hfill {\small(Prop.~\ref{prop:comp_output_lower})}\\ \cmidrule{2-4}
 & \sf\textbf{Upper} & $\oof{\card{\onti}\cdot\left(\card{\ont}^2+\card{\onti}\right)}$ \hfill {\small(Cor.~\ref{cor:comp_crowd_upper})}& $\oof{\card{\onti}\cdot\left(\card{\ont}^2+\card{\mfi}+\card{\mii}\right)}$  \hfill {\small (Prop.~\ref{prop:upper_comp_output})}\\
\bottomrule

\end{tabularx}
\vspace*{-2mm}
\caption{Summary of the main complexity results, where we have $\card{\onti} \leq 2^{\oof{\card{\ont}}}$ and $\card{\onts}\leq 2^{\oof{\card{\onti}}}$}
\label{tbl:summary}
\end{table*}

\paragraph*{Results}
For our theoretical results, we harness tools from three areas of
computer science: \emph{data mining}, \emph{order theory} and \emph{Boolean function
learning}~\cite{bioch1995complexity,bshouty1995exact,gainanov1984one,gunopulos2003discovering,linial1985every,mannila1997levelwise,srikant1995mining}.
Order theory is relevant to our discussion, because a taxonomy is in fact a partial order over data items; and Boolean function learning is relevant since the set of frequent itemsets to identify can be represented as a Boolean
function indicating whether itemsets are frequent, a connection that was also pointed out in previous works in data mining~\cite{mannila1997levelwise}. Our contribution in this paper is combining and extending these tools to
characterize
the complexity of crowd mining.

A summary of our main results is presented in Table~\ref{tbl:summary}, where we
give upper and lower bounds for our two complexity measures. In the first column, we give such bounds as a function of the structure of the input taxonomy $\ont$. These bounds are not affected by properties of the output, such as the actual number of frequent itemsets to be identified. In contrast, in the second column, we give complexity bounds as a function of the number of maximal frequent itemsets (MFIs) and minimal infrequent itemsets (MIIs). Intuitively, the MFIs and MIIs (to be defined formally later) are alternative concise descriptions of the frequent itemsets, and thus capture the output of the mining process.

The first row of Table~\ref{tbl:summary} presents crowd complexity results. We show that, given a taxonomy $\ont$, the problem of identifying the frequent itemsets has a tight bound  logarithmic in $\card{\onts}$~-- the number of possible Boolean frequency functions, which depends on $\ont$. As
reflected in the inequalities at the bottom of Table~\ref{tbl:summary} (and explained in Section~\ref{sec:query}), $\log\card{\onts}$ is at most exponential in $\card{\ont}$. When the output is considered, our lower complexity bound is the sum of the numbers of MFIs and MIIs, and the
upper bound adds the taxonomy size as a multiplicative factor. We provide a constructive algorithm
(Algorithm~\ref{alg:mfi_upper})
that achieves this bound.

In the second row of the table, we study computational complexity. We focus on ``crowd-efficient'' algorithms, which achieve the crowd complexity upper bound mentioned above. The crowd complexity lower bound is trivially a lower bound of computational complexity, but w.r.t.\ the output we obtain a stronger hardness result by showing that the problem is \emph{EQ-hard} in the taxonomy size and in the numbers of MFIs and MIIs. EQ is a basic problem in Boolean function learning, not known to be solvable in PTIME~\cite{bioch1995complexity,gottlob2013deciding}. As for upper bounds, from
Algorithm~\ref{alg:mfi_upper}, we obtain an upper computational bound polynomial in $\card{\onti}$ -- the number of (relevant) itemsets of $\ont$. This size is at most exponential in $\card{\ont}$
(see Section~\ref{sec:bg}). Algorithm~\ref{alg:mfi_upper} is not crowd-efficient w.r.t.\ the input alone, but for the upper computational complexity bound we relax this requirement in order to achieve a more feasible bound.
Omitted from Table~\ref{tbl:summary} are results for the case in which the size of itemsets is bounded by a constant $k$, which we also study for the problem axes mentioned above.

Finally, given the relatively high lower complexity bounds, we examine two additional approaches. The \emph{chain partitioning} approach, following a standard technique in data mining and Boolean function learning, suggests an alternative algorithm for crowd mining. We show that while this algorithm is not crowd-efficient in general, it outperforms
Algorithm~\ref{alg:mfi_upper}
given certain conditions on the frequent itemsets.
The \emph{greedy} approach attempts to maximize, at each question to the crowd, the number of \emph{itemsets} that are classified as frequent or infrequent. We show that choosing a question that maximizes this number is $\mathrm{FP}^{\mathrm{\#P}}$-hard in $\card{\ont}$.

\paragraph*{Paper organization} We start in Section \ref{sec:bg} by formally defining the setting and the problem. Crowd and computational complexity are studied in Sections~\ref{sec:query} and~\ref{sec:computational} respectively. We consider chain partitioning in Section~\ref{sec:relax}, and a greedy approach in Section~\ref{sec:greedy}. Related work is discussed in Section~\ref{sec:related} and we conclude in Section~\ref{sec:conc}.

%% file: prelim.tex
\newpage
\section{Preliminaries}
\label{sec:bg}

We now present the formal model and problem settings
for taxonomy-based  crowd mining. Table~\ref{tab:notations} summarizes all the introduced notation. We start by recalling some basic
itemset mining definitions from~\cite{agarwal1994fast} and explain
how they apply to our settings.

Let $\idom{}=\{i_1,i_2,i_3,\dots\}$ be a finite set of distinct item names. Define an \emph{itemset} (or \emph{transaction}) $A$
as a subset of $\idom{}$. Define a
\emph{database} $D$
as a bag (multiset) of transactions. $\card{D}$ denotes the number
of transactions in $D$. The frequency or \emph{support} of an
itemset $A\subseteq\idom{}$ in $D$ is
\(\supp{D}{A}\colonequals \card{\{T\in D\mid A\subseteq
T\}}/\card{D}\). $A$ is considered \emph{frequent} if its support
exceeds a predefined threshold
$\Theta$: we assume that $0<\Theta<1$ as mining is trivial when $\Theta \in \{0, 1\}$. Given a database $D$, we define the predicate $\freqf$ which takes an itemset as
input and returns \true{} iff this itemset is frequent in $D$ (the dependency
on $D$ is omitted from the notation).

For example, in the domain of leisure activities, $\idom{}$ may include different activities, relevant equipment, locations, etc. Each transaction $T$ may represent all the items involved in a particular leisure event (a vacation day, a night out, etc.). If, e.g., the set \{\itm{tennis}, \itm{racket}, \itm{sunglasses}\} is frequent, it means that a \txtitm{racket} and \txtitm{sunglasses} are commonly used for \txtitm{tennis}. Or, if
\{\itm{indoor\_cycling}, \itm{TV}\} is frequent, it may imply that indoor cyclists often watch \txtitm{TV}
while cycling.

In our crowd-based setting, the database of interest $D$ is not materialized and only models the knowledge of
people, so we can only access $D$ by asking them questions. As shown in~\cite{amsterdamer2013crowd},
we can ask people for summaries of their personal knowledge, which we can then interpret as data patterns -- itemset frequencies
in our case. We thus abstractly model a \emph{crowd query} as follows:

\begin{definition}[Crowd query]
A \emph{crowd query} takes as input an itemset $A\subseteq\idom{}$ and returns
$\freqf[A]$.
\end{definition}

When using crowdsourcing to answer this type of crowd queries, and when posing questions to the crowd in general, one must deal with imprecise or partial answers. This general problem was studied in previous crowdsourcing works~\cite{amsterdamer2013crowd,amsterdamer2013crowdminer,parameswaran2013crowdscreen}. We can employ one of their methods as a black-box and assume that each crowd query is posed to a sufficient (constant) number of users, so as to gain sufficient confidence in the obtained Boolean answer.
Thus, the cost of a crowd mining algorithm can be defined as \emph{the number of crowd queries} rather than the number of posed questions:
the crowd acts as an \emph{oracle} for itemset frequency. The cost metric does not depend on the size of the hypothetic database $D$, or the number of scans that would be necessary to determine the frequency of the queried itemsets if $D$ were materialized.

\paragraph*{Itemset Dependency and Taxonomies}
The support of different itemsets can be dependent. For example, if $A\subseteq B$ then $B\subseteq T$ implies $A\subseteq T$ so \(\supp{D}{A}\geq\supp{D}{B}\) for every $D$. This fundamental property is used by classic mining algorithms such as Apriori~\cite{agarwal1994fast}.

Moreover, as noted in~\cite{srikant1995mining}, there may be dependencies between itemsets resulting from \emph{semantic relations} between items. For instance, in our example from the Introduction, the itemset \{\itm{sunglasses}, \itm{sport}\} is semantically implied by any transaction containing \{\itm{sunglasses}, \itm{tennis}\}, since tennis is a sport.

Such semantic dependencies can be naturally captured by a \emph{taxonomy}~\cite{srikant1995mining}. Formally, we define a taxonomy as a partially ordered set (or \emph{poset}) $\ont{}=(\idomo{}, \orelt)$ where $\orelt$ is a partial order over the element domain $\idomo$. $i\orelt{} i'$ indicates that item $i'$ is more specific than $i$ (any $i'$ is also an $i$).
Observe that the antisymmetry of $\orelt$ implies that no two different items in $\idomo{}$ are equivalent by $\orelt$.\footnote{Such equivalence would stand for semantic synonyms such as \itm{cycling} and \itm{biking}. We thus assume that every group of synonyms is represented by a single item.} We use $i < i'$ when $i \leq i'$ and $i \neq i'$, and denote by $\orel{}$ the \emph{covering relation} of $\orelt$: $i \orel{} i'$ iff $i<i'$ and there exists no $i''$ s.t.\ $i<i''<i'$.

We represent posets as DAGs, whose vertices are the poset elements, and where a directed edge $(i,i')$ indicates that $i \orel{} i'$. This is in line with standard representations of
posets such as, e.g., \emph{Hasse diagrams}. We denote by $\card{\ont} = \oofi{\card{\idomo}^2}$ the size of the taxonomy including the number of elements and pairs in $\orel$.

\begin{figure*}
\hspace{-2.5ex}
\begin{subfigure}[b]{0.09\textwidth}
  \centering
  \noindent \includegraphics[scale=0.28]{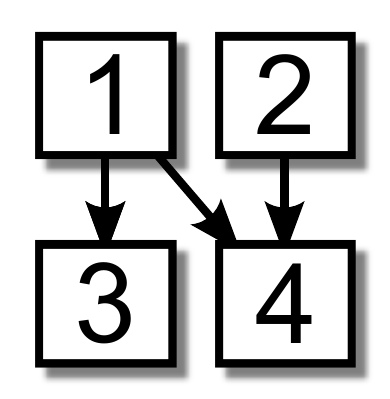}
  \subcaption{$\ont_1$}
  \label{subfig:ont1}
\end{subfigure}
\hfill
\begin{subfigure}[b]{0.1\textwidth}
  \centering
  \includegraphics[scale=0.28]{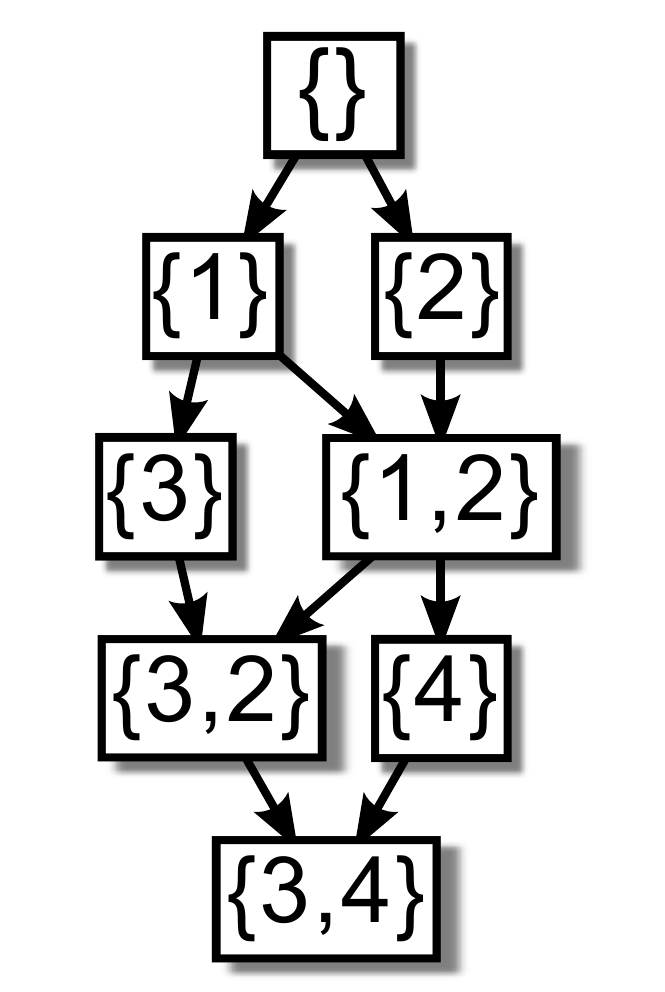}
  \subcaption{$\ifunct{\ont_1}$}
  \label{subfig:onti1}
\end{subfigure}
\hfill
\begin{subfigure}[b]{0.163\textwidth}
  \centering
  \includegraphics[scale=0.28]{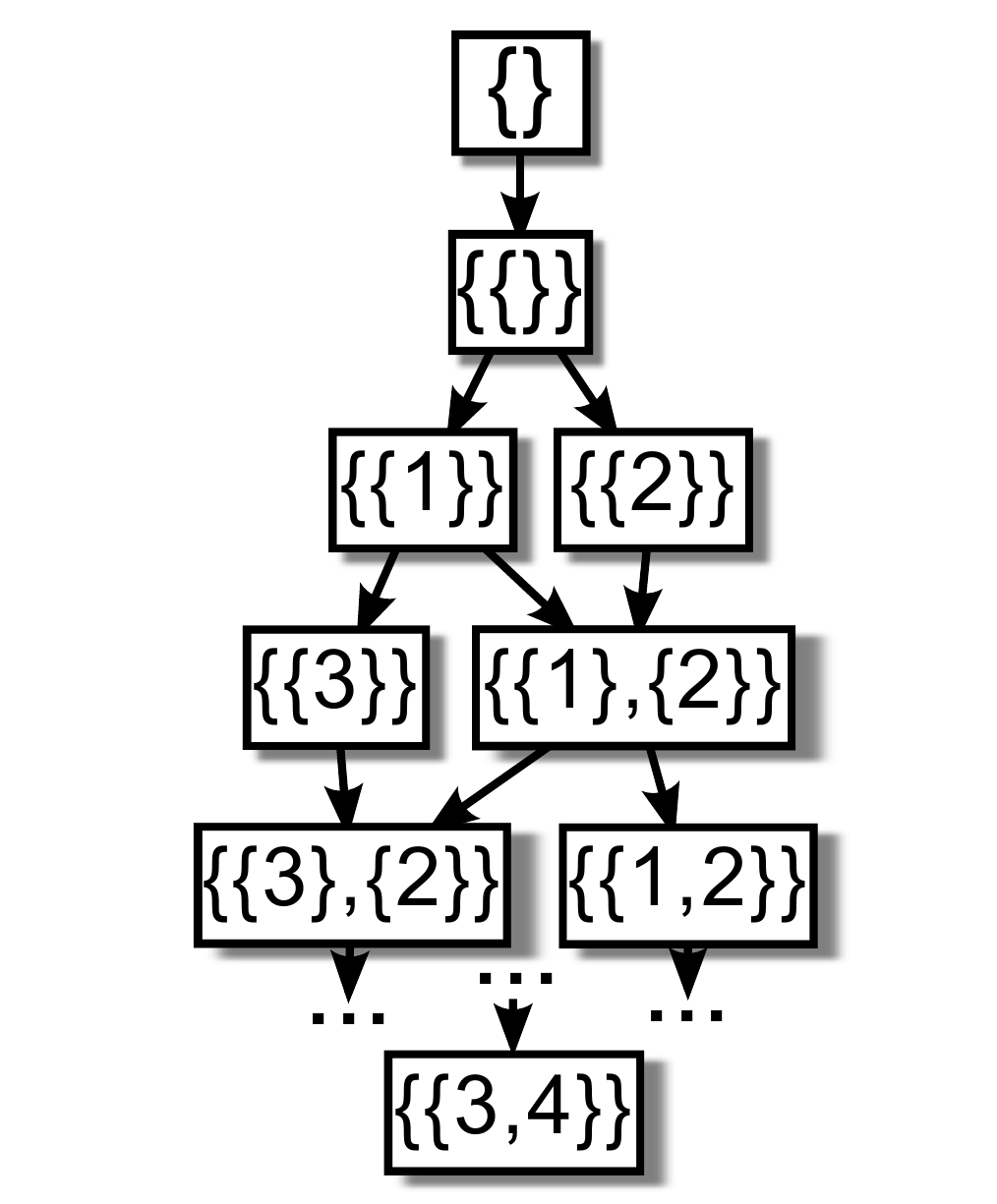}
  \subcaption{$\funct{S}{\ont_1}$}
  \label{subfig:onts1}
\end{subfigure}
\hfill
\begin{subfigure}[b]{0.09\textwidth}
  \centering
  \includegraphics[scale=0.28]{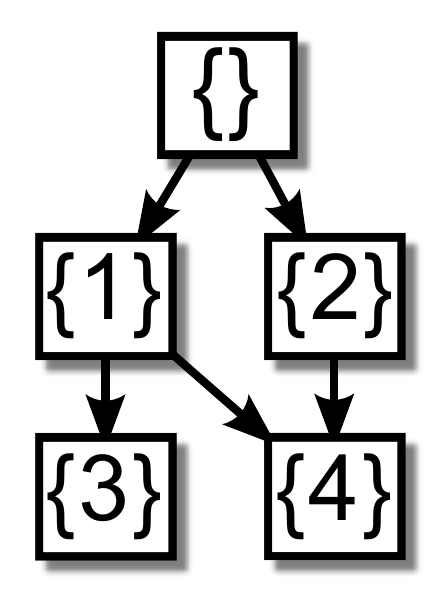}
  \subcaption{$\ikfunct{1}{\ont_1}$}
  \label{subfig:ontik}
\end{subfigure}
\hfill
\begin{subfigure}[b]{0.13\textwidth}
  \centering\vspace{0pt}
  \parbox{1.7cm}{\centering\includegraphics[scale=0.28]{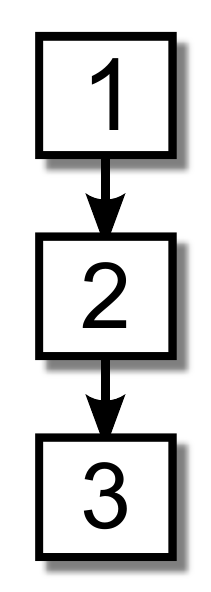}}
  \subcaption{$\ont_2$ -- ``chain''}
  \label{subfig:chain}
\end{subfigure}
\hfill
\begin{subfigure}[b]{0.07\textwidth}
  \centering
  \includegraphics[scale=0.28]{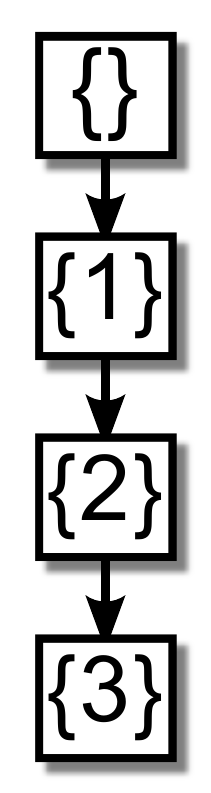}
  \subcaption{$\ifunct{\ont_2}$}
  \label{subfig:chaini}
\end{subfigure}
\hfill
\begin{subfigure}[b]{0.13\textwidth}
  \centering\vspace{0pt}
  \includegraphics[scale=0.28]{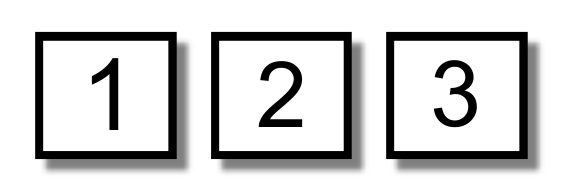}
  \subcaption{$\ont_3$ -- ``flat''}
  \label{subfig:flat}
\end{subfigure}
\hfill
\begin{subfigure}[b]{0.195\textwidth}
  \centering
  \includegraphics[scale=0.28]{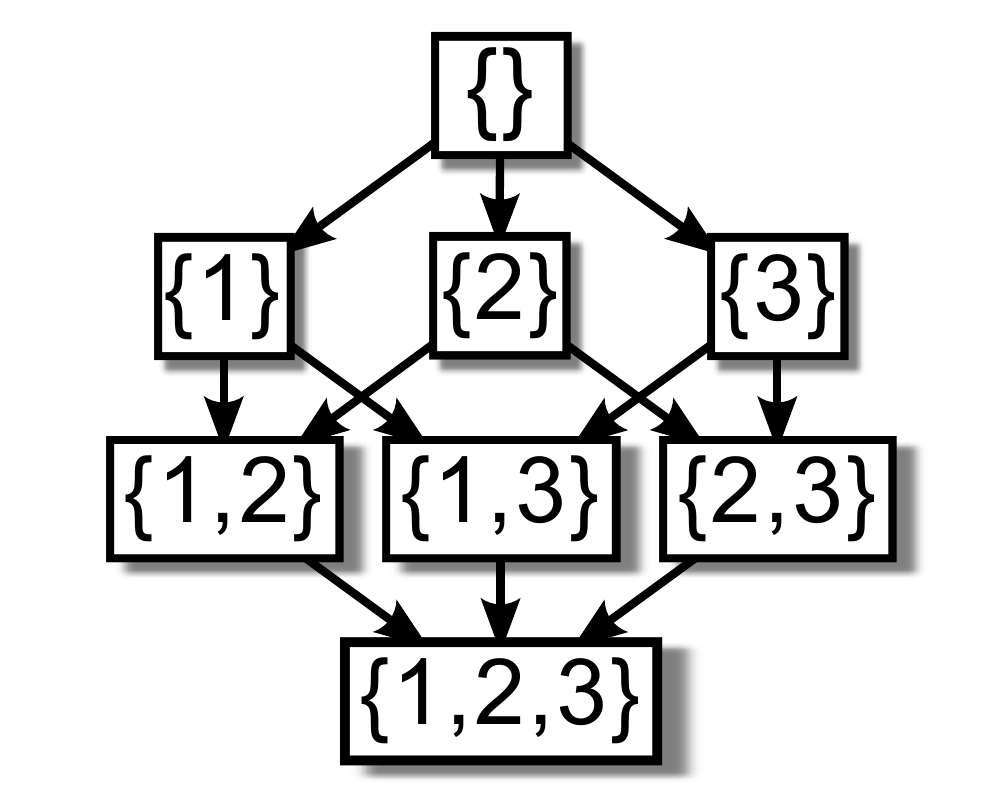}
  \subcaption{$\ifunct{\ont_3}$ -- Boolean lattice}
  \label{subfig:flati}
\end{subfigure}
\hspace{-0.68em}
\vspace{-2mm}
\caption{Example taxonomies}
\label{fig:ex1}
\end{figure*}

\begin{example}
Consider the taxonomy $\ont_1$ shown in Figure~\ref{subfig:ont1}. We can label its elements with items, e.g.:
\begin{inparaenum}
\item \itm{cycling},
\item \itm{sport}, \item \itm{bicycle\_touring},
\item \itm{indoor\_cycling}.
\end{inparaenum}
The interpretation of the taxonomy would then be: both \txtitm{bicycle touring} and \txtitm{indoor cycling} are types of \txtitm{cycling}, and \txtitm{indoor cycling} is also a \txtitm{sport}.
\end{example}

Let us briefly define some general useful terms in the context of
posets. When $i \leq i'$ we call $i$ an \emph{ancestor} and $i'$ a
\emph{descendant}. Similarly, when $i\lessdot i'$ we call $i$ a
\emph{parent} and $i'$ its \emph{child}. A \emph{chain} is a sequence of elements $i_1< i_2<\dots< i_n$. An \emph{antichain}
is a set of elements $A=\{i_1,\dots,i_n\}$ that are
incomparable with respect to $\leq$, i.e., there exist no $i_j\neq
i_k\in A$ s.t.\ $i_j\leq i_k$. The \emph{width} of a poset $P$,
denoted by $\wid[P]$, is the size of its largest antichain.
An \emph{order ideal} (or lower set) $A$ of a poset $P$ is a subset of its
elements s.t.\ if $i \in A$ then all the ancestors of $i$ are in $A$.

\begin{example}
The antichains of the example taxonomy $\ont_1$ include the empty
antichain $\{\}$; singleton itemsets such as $\{3\}$; and antichains
of size~2 such as $\{2,3\}$ (since~2 and~3 are incomparable). There
are no larger antichains, and thus $\wid[\ont_1]=2$.
\end{example}

We denote by $\idomop{}$ the domain of antichains of elements from $\ont$. Antichains are concise in the sense that they contain no items implied by other
items. In this way, e.g., \{\itm{tennis}\} concisely represents \{\itm{tennis}, \itm{sport}\}, \{\itm{tennis}, \itm{sport}, \itm{activity}\}, etc. Thus, unless stated otherwise, whenever we
mention itemsets we assume that the items form an antichain. This is also useful for practical purposes: it would be strange, e.g., to ask users whether they simultaneously play \txtitm{tennis} and do \txtitm{sport}.

Based on $\orelt$, the semantic relationship between items, we can define the corresponding relationship between \emph{itemsets}. For itemsets $A,B$ we define $A\orelit
B$ iff every item in A is implied by some item in $B$. Formally:

\begin{definition}[Itemset taxonomy]
\label{def:onti}
~~Given a taxonomy $\ont{}=(\idomo,\orelt)$ define
its \emph{itemset taxonomy} as the poset $\onti{}=(\idomop,
\orelit)$. By an abuse of notation, we extend $\orelit$ to itemsets,
where for every two itemsets $A,B\in\idomop$, $A\orelit B$ iff
$\forall i\in A, \exists i'\in B~ i\orelt{} i'$.
Similarly, we extend $<$ and $\orel$ to itemsets: $A < B$ when $A \leq B$ and $A \neq B$, and $A \orel{} B$ iff
$A<B$ and there exists no $C$ s.t.\ $A<C<B$.
\end{definition}

Figure~\ref{subfig:onti1} illustrates $\ifunct{\ont_1}$, the itemset taxonomy of $\ont_1$ from Figure~\ref{subfig:ont1}. Observe that, for singleton itemsets, $\orelit$  corresponds to the order on items.

Finally, we redefine support to take $\onti$ into account.

\begin{definition}[Itemset support]
Let $A\subseteq\idom{}$ be an itemset. We define the support of $A$ w.r.t.\ a database $D$ and a taxonomy $\ont{}$ to be~0 if $D$ is empty, and as \(\supp{D,\ont}{A}\colonequals\card{\{T\in D\mid A\orelit{} T\}}/\card{D}\) otherwise.
\end{definition}

\paragraph*{Properties of the itemset taxonomy}
By construction, $\onti{}$ is not an arbitrary poset: for instance, it always has a single ``root'' element, namely the empty itemset, which precedes all other elements by $\orelit$.
More generally, the domain of all possible itemset
taxonomies can be precisely characterized as the domain of all \emph{distributive lattices} (see Appendix~\ref{sec:struct} for details). We illustrate the structure of $\onti$ in two extreme but useful examples.

\begin{example}
\label{ex:ontis}
Figure~\ref{subfig:chain} illustrates a total order or ``chain'' taxonomy,
whose itemset taxonomy is a chain with one more element (Figure~\ref{subfig:chaini}). Figure~\ref{subfig:flat} displays a ``flat'' taxonomy, where all the elements are incomparable. Its itemset taxonomy (Figure~\ref{subfig:flati}) contains all the possible itemsets: it is the Boolean lattice structure explored by classic data mining algorithms such as Apriori~\cite{agarwal1994fast}. Hence, if a
flat $\ont$ has $n$ elements, $\onti$ has exactly $2^n$ elements
and $\orelt$ corresponds exactly to set inclusion.
\end{example}

\paragraph*{Maximal frequent itemsets}
By the definition of support, $\freq$ is a \emph{(decreasing) monotone predicate} over itemsets, i.e., if $A\orelit B$ then $\freqf[B]$ implies $\freqf[A]$.
Consequently, $\freq$ can be uniquely and concisely characterized by a set of \emph{maximal frequent itemsets} (MFIs), namely all the frequent itemsets with no frequent descendants. Equivalently, it can be characterized by a set of \emph{minimal infrequent itemsets} (MIIs), which are all the infrequent itemsets with no infrequent ancestors.
MFIs and MIIs were introduced for knowledge discovery in~\cite{mannila1997levelwise} (where they are called respectively the \emph{positive border} and \emph{negative border}),
and existing data mining algorithms such as~\cite{bayardo1998efficiently} try to identify them as a concise representation of the frequent itemsets.
We denote the MFIs and MIIs of a predicate $\freq$ by $\mfi$ and $\mii$ respectively, where $\freq$ is clear from the context.
More generally, we call \emph{maximal elements} the analogue of MFIs for decreasing monotone
predicates over an arbitrary poset.

\begin{example}
Consider $\ifunct{\ont_1}$ in Figure~\ref{subfig:onti1}. Assume, e.g., that we know $\freqf[\{2\}]=\freqf[\{3\}]=\freqf[\{4\}]=\true{}$. By the monotonicity of $\freq$, their ancestors (e.g., $\{1,2\}$) are also frequent. Assume that $\freq$ returns \false{} for any other itemset. Then $\freq$ can be uniquely characterized by
its MFIs $\{3\}$ and $\{4\}$, or by its MII $\{3,2\}$: the values of
$\freq$ for the other itemsets follow.
\end{example}

\input{notations}

\paragraph*{Restricting the itemset size}
In typical crowd scenarios, there are often restrictions on the number of
elements that may be presented in a crowd query~\cite{marcus2012counting},
so it is impractical to ask users about very large itemsets. We
therefore define a variant of the problem in which the itemset size is
bounded from above by a constant.

\begin{definition}[$k$-itemset taxonomy]
We define the \emph{$k$-itemset taxonomy} $\ontik\colonequals(\idomok, \orelikt)$ where $\idomok \colonequals\{A\in\idomop\mid\card{A}\leq k\}$ and $k$ is constant.
\end{definition}

We refer to the elements of $\idomok$ as \emph{$k$-itemsets}. Observe that, unlike for $\onti$, the number of itemsets in $\ontik$ is always polynomial in $\card{\idomo{}}$, i.e.,
$\oofi{\card{\idomo}^k}$. In addition, $\ontik$
need not be a distributive lattice: for instance, by setting $k=1$, $\ontik$ is almost identical to $\ont$, i.e., it is an arbitrary poset except for the added $\{\}$ element. Compare for example $\ikfunct{1}{\ont_1}$ (Figure~\ref{subfig:ontik}) with $\ont_1$ (Figure~\ref{subfig:ont1}).

\paragraph*{Problem statement}
Given a \emph{known} taxonomy $\ont{}$ and an \emph{unknown} database $D$, defining the (also unknown) frequency predicate $\freq$ over the itemsets in $\onti{}$, we denote by \aprob{} the problem of identifying, using only crowd queries, \emph{all} the frequent itemsets in $D$ (or, equivalently, of identifying $\freq$
exactly). We consider \emph{interactive} algorithms that iteratively compute, based on the knowledge collected so far, which crowd query to pose next, until \aprob{} is solved.

As mentioned in the Introduction, we study the complexity bounds of such algorithms for two metrics. We first consider the number of crowd queries that need to be asked, namely the \emph{crowd complexity}. Then, we study the feasibility of ``crowd-efficient'' algorithms, by considering the \emph{computational complexity} of algorithms that achieve the upper crowd complexity bound. This last restriction is relaxed in the sequel.

%% file: notations.tex
\begin{table}
\begin{tabularx}{\columnwidth}{Xl}
  \toprule
  $\supp{D}{A}$ & Support of itemset $A$ in database $D$\\
  $\Theta$ & Support threshold, $0 < \Theta < 1$\\
  $\freqf[A]$ & True iff $\supp{D}{A}$ exceeds $\Theta$\\
  \midrule
  $\idom$ & Set of all items\\
  $\ont$ & Taxonomy -- a partial order $(\idomo, \orelt)$\\
    $i \orelt i'$ & Item $i'$ is more specific than $i$\\
  $\orel$ & Covering relation of $\orelt$\\
    $\card{\ont}$ & Size of the taxonomy (as the DAG of $\orel$)\\
  $\wid[\ont]$ & Width of $\ont$ \\
  \midrule
  $\idomop$ & Antichains of $\ont$\\
  $\onti$ & Itemset taxonomy $(\idomop, \orelit)$\\
  $A \orelit B$ & $\forall i \in A, \exists i' \in B~ i \orelt i'$\\
  $\idomok$ & Itemsets of size $\leq k$\\
  $\ontik$ & $k$-itemset taxonomy $(\idomok, \orelikt)$\\
  \midrule
  $\mfi$ & Maximal frequent itemsets\\
  $\mii$ & Minimal infrequent itemsets\\
  $\onts$ & Solution taxonomy $\ifunct{\onti}$\\
  $\ontsk$ & Solution taxonomy $\ifunct{\ontik}$\\
  $\aprob$ & Problem of identifying $\freq$ exactly\\
\bottomrule
\end{tabularx}
\caption{Summary of notations}
\label{tab:notations}
\end{table}

%% file: query-complexity.tex
\needspace{4\baselineskip}
\section{Crowd Complexity}
\label{sec:query}
We now analyze the crowd complexity of \aprob{}, first w.r.t.\ the input taxonomy. Then, we consider the complexity w.r.t.\ the output, which allows for a finer analysis depending on properties of $\freq$. As a general remark for our analysis, note that we can always avoid querying the same itemset twice, e.g., by caching query answers. Thus, every upper bound $\oof{\name{X}}$ presented in this section
is actually $\oof{\name{min}\{\name{X},\card{\onti}\}}$.

\needspace{2em}
\subsection{With Respect to the Input}
\label{sec:query-wc}
To illustrate the problem boundaries, consider the following specific cases of $\ont{}$ for which we know the optimal solution strategy. For a chain taxonomy (as in Figure~\ref{subfig:chain}), identifying $\freq$ amounts to a binary search for the single MFI. This can be done in $\oof{\log\card{\idomo}}$ steps. For a flat taxonomy (as in Figure~\ref{subfig:flat}), for which the elements of $\onti$ are the power set of $\idomo$, identifying $\freq$ is equivalent to learning a monotone Boolean function over $n$ variables, where $n=\card{\idomo}$. For this problem, a tight bound of $\Theta(\sfrac{2^n}{\sqrt{n}})$ is known~\cite{hansel1966nombre,korobkov1965estimating}.

We study the solution for a \emph{general taxonomy structure}. Let us start by
defining the following:

\begin{definition}
  \emph{(\textsc{Solution taxonomy})}.
\label{def:onts}
Given a taxonomy $\ont$, we define its \emph{solution taxonomy} $\onts{}=\ifunct{\onti}$. The domain of elements of $\onts{}$ is $\prop{\onti}$, i.e., antichains of itemsets.
\end{definition}

We call this construction the solution taxonomy, since its elements correspond, precisely, to the frequency predicates over $\onti$, i.e., all
possible solutions to \aprob{} for a given $\ont$. More precisely, each element of $\onts$ is an antichain of itemsets that is exactly the set of MFIs $\mfi$ of some $\freq$ predicate. We prove this below but first illustrate the structure via an example.

\begin{example}
\label{ex:onts}
Figure~\ref{subfig:onts1} illustrates parts of the solution taxonomy of the running example, $\ont_1$ (not provided fully due to its size). Consider, e.g., $\{\{1\},\{2\}\}$. This element of $\funct{S}{\ont_1}$ corresponds to the $\freq$ predicate assigning \true{} (only) to $\{1\}$, $\{2\}$ and their ancestor $\{\}$ in $\ifunct{\ont_1}$. Similarly, $\{\}$ in $\funct{S}{\ont_1}$ corresponds to a predicate assigning \false{} to every itemset, and $\{\{3,4\}\}$ to the one assigning \true{} to every itemset.
\end{example}

Next, we prove the claim about the correspondence between $\onts$ elements and frequency predicates by first showing a bijective correspondence between elements of $\onts$ and
(decreasing) monotone predicates over $\onti$. Then, we show that each such predicate can indeed serve as a frequency predicate for some database.

\begin{proposition}
\label{prop:onts_predicates}
There exists a bijective mapping from $\prop{\onti}$ to the monotone predicates over $\onti$.
\end{proposition}

The proof maps every predicate in a general poset to its (unique) set of maximal elements, which necessarily forms an antichain. Thus, every predicate over $\onti$ can be mapped to an antichain of itemsets, which is
an element of $\onts$ (see formal proof in Appendix~\ref{sec:crowd_app}).

\begin{proposition}
For every threshold $0 < \Theta < 1$, every monotone predicate $F$ over $\onti$
is the frequency predicate $\freq$ of some
database $D$.
\end{proposition}

\begin{proof}
Given $F$, let $\mathcal{M}$ be its set of maximal elements.
If $\mathcal{M}$ is empty, i.e., no itemset should be frequent, $F$ is realized by the empty database. Otherwise, construct $D$ to consist of the following $d$ transactions: $n$ ``full'' transactions with all the items of $\idomo$, one transaction per $A \in \mathcal{M}$ containing exactly $A$, and $d - n - \card{\mathcal{M}}$ empty transactions.
We choose $d$ and $n$ s.t.\ every (non-empty) itemset $B$ is frequent in $D$ iff it is supported by $>n$ transactions, or, equivalently, iff it is supported by at least one of the $\card{\mathcal{M}}$ non-trivial transactions. To do that, pick an integer $d$ that is large enough such that there exists an integer $n$ such that $0 \leq \sfrac{n}{d} < \Theta < \sfrac{(n+1)}{d} < \sfrac{(n + \card{\mathcal{M}})}{d} \leq 1$. When this holds, the frequent itemsets of $D$ are exactly the ancestors of itemsets in $\mathcal{M}$, so $\freq = F$ as desired.
\end{proof}

We have now shown that the \aprob{} problem (identifying $\freq$) is equivalent to finding the element of $\onts$ corresponding to
$\freq$, namely $\mfi$. Before we study the complexity of this last task, let us first describe abstractly how the solutions space is narrowed down during the execution of \emph{any} algorithm that solves \aprob{}.
In the beginning of the execution, all the elements of $\onts$ are \emph{possible}.
The algorithm uses some decision method to pick an itemset $A\in\prop{\ont}$ and queries it.
If the answer is \true{} ($A$ is frequent), this means that
$\mfi$ contains $A$ or one of its descendants in
$\onti$, so we can eliminate all MFI sets of $\onts$ that do not have this property, which we can show are exactly the non-descendants of $\{A\}$ in $\onts$.
Conversely, if the answer is \false{} ($A$ is infrequent),
we can eliminate all the descendants of $\{A\}$ in $\onts$ (including $\{A\}$). The last remaining
element in $\onts$ at the end corresponds to the correct $\freq$ predicate, because it is the only one consistent with the observations.

\begin{example}
Consider again $\funct{S}{\ont_1}$
in Example~\ref{ex:onts}. By discovering that, e.g., $\freqf[\{1\}]$, we know
e.g.: that $\{2\}$ cannot be the only MFI
so we can eliminate the solution element $\{\{2\}\}$, that its ancestor $\{\}$ is not an MFI
so we can eliminate
$\{\{\}\}$, and so on.
In total, all non-descendants of $\{\{1\}\}$ in $\funct{S}{\ont_1}$ can be eliminated.
\end{example}

\paragraph*{Lower Bound}
We now give a lower crowd complexity bound for solving \aprob{} in terms of the input, which is proved to be tight in the sequel.
The proof relies on the fact that solving \aprob{} amounts to searching for an element in $\onts$; it can be given as a simple information-theoretic argument (see Appendix~\ref{sec:crowd_app}), but we present it in connection with~\cite{linial1985every} as we will reuse this link to obtain our upper bound.

\begin{proposition}
\label{prop:lower-query}
The worst-case crowd complexity of identifying $\freq$ is $\omegaof{\log\left(\card{\onts}\right)}$.
\end{proposition}
\begin{proof}
This is implied by the analogous claim of~\cite{linial1985every} about order ideals in general posets. By characterizing an order ideal by its maximal elements (whose descendants are not in the ideal) we obtain an antichain, which defines a bijective correspondence between order ideals and antichains. Thus, we can map antichains to order ideals, and use the observation in~\cite{linial1985every} directly to obtain the same lower bound.
\end{proof}

In the worst case, $\log\card{\onts}$ can be linear in $\card{\onti}$, which itself may be exponential in $\card{\ont}$ (e.g., for a flat taxonomy). When this is the case, a trivial algorithm achieves the complexity bound by querying every element in $\onti$.
However, for some taxonomy structures (e.g., chain taxonomies), the size of $\onts$ is much smaller.
We now use $\wid[\onti]$ to deduce a more explicit lower bound for Prop.~\ref{prop:lower-query}.

\begin{proposition}
$\wid[\onti]  \geq {\wid \choose \lfloor \wid/2 \rfloor}$.
\end{proposition}
\begin{proof}
By definition, there exists at least one itemset in $\onti$ of size $\wid$. This itemset has $\wid \choose \lfloor \wid/2 \rfloor$ subsets of size $\lfloor\wid/2\rfloor$. These itemsets are also in $\onti$, since they only contain incomparable items. Moreover,
they are pairwise incomparable in $\onti$, and thus form an antichain
whose size yields the lower bound.
\end{proof}

By replacing $\ont$ with $\onti$ we get a lower bound for $\wid[\onts]$. We can thus prove the following bound which, though weaker, is more explicit than Prop.~\ref{prop:lower-query} as it is expressed in terms of the original ontology width rather than $\card{\onts}$.

\begin{corollary}
The worst-case crowd complexity of identifying $\freq$ is $\omegaof{\sfrac{2^{\wid}}{\sqrt{\wid}}}$.
\end{corollary}
\begin{proof}
  $\card{\onts} > \wid[\onts] \geq {\wid[\onti] \choose \lfloor \wid[\onti]/2 \rfloor}$. We obtain \(\log\card{\onts} \geq \omegaof{\log\left(\sfrac{2^{\wid[\onti]}}{\sqrt{\wid[\onti]}}\right)} = \omegaof{\wid[\onti]}\) using Stirling's approximation; and finally, using the lower bound of $\wid[\onti]$ and applying the approximation again, we express the bound in terms of $\wid$.
\end{proof}

\paragraph*{Upper Bound}
We now state a tight upper bound (i.e., that matches the lower bound up to a multiplicative constant). The proof relies on Theorem~1.1 of~\cite{linial1985every}, which shows that, in any poset, there exists an element
such that the proportion of order ideals (or, in our case, elements of $\onts$) that contain the element is within a constant range of $1/2$.
Hence, a greedy strategy that queries
such elements
will eliminate a constant fraction of the possible solutions at each step and
completes in a time that is logarithmic in the size of the search space. The proof details are deferred to Appendix~\ref{sec:crowd_app}.

\begin{proposition}
\label{prop:upper-crowd-input}
The worst-case crowd complexity of identifying $\freq$ is $\oof{\log\card{\onts}}$.
\end{proposition}

\subsection{With Respect to the Input and Output}
\label{sec:crowd_output}
So far our results only relied on the structure and size of the input taxonomy $\ont$. However, as noted in Section~\ref{sec:bg}, the characteristics of the output $\freq$ predicate may have a crucial effect on the problem complexity, because, in practical scenarios, the number of MFIs and MIIs is usually small. For instance, when dealing with leisure habits, the number of activities that are commonly performed together in the population is typically very small w.r.t.\ all the combinations that the taxonomy allows.
Hence, we next study the effect of the output on the crowd complexity boundaries of \aprob{}.

\paragraph*{Lower Bound}
Since each of the sets of MFIs and MIIs uniquely represents the $\freq$ predicate, one could hope that it would be sufficient to identify only one of them
to solve \aprob{}. However, it turns out that
one must query at least all the MIIs to verify that the MFIs are maximal,
and vice versa.
This result is well-known for Boolean lattices~\cite{gainanov1984one};
in our setting it follows from
the more general
Thm.~2 of~\cite{mannila1997levelwise} (which concerns \emph{any} partial order rather than just distributive lattices).

\begin{proposition}
\label{prop:mfi_lower}
The worst-case crowd complexity of identifying $\freq$ is $\omegaofi{\card{\mfi} + \card{\mii}}$.
\end{proposition}

Hence, though we can describe the output by its set of MFIs (or
MIIs), we need to query both the MFIs and MIIs. This implies that the crowd complexity may be exponential even in the minimal output size, since the difference between $\card{\mfi}$ and $\card{\mii}$ may be large though only one
suffices to describe the output. This is derived from a known result in Boolean function learning~\cite{bshouty1995exact}.

\begin{corollary}
\label{prop:mfi_lower_exp}
\emph{(\textsc{see \cite{bshouty1995exact}})}.
\hspace{-0.1em}The worst-case crowd complexity of identifying $\freq$ is $\omegaof{2^{\min\{\card{\mfi},\card{\mii}\}}}$
\end{corollary}
We note that the current lower bound is not tight: for instance, over a chain taxonomy,
$\card{\mfi} + \card{\mii} \leq 2$ for any $\freq$ predicate, but we already noted
in Section~\ref{sec:query-wc} that the worst-case crowd complexity in this case is $\omegaof{\log \card{\idomo}}$.

\paragraph*{Upper bound}
We next show an upper bound that is
within a factor $\card{\idomo}$ of the lower
bound of Prop.~\ref{prop:mfi_lower}.
It generalizes known MFI and MII identification algorithms for the case where there is no underlying taxonomy, such as the monotone Boolean function learning algorithm of~\cite{gainanov1984one} and the Dualize and Advance algorithm of~\cite{gunopulos2003discovering,gunopulos1997data}; see Section~\ref{sec:related} for an in-depth comparison.
Intuitively, our algorithm traverses the elements of $\onti$ in an efficient way to identify an MFI or an MII, and repeats this process as long as there are \textsf{unclassified} elements in
$\onti$,
i.e., elements that are not known to be \textsf{frequent} or \textsf{infrequent}. Using this method we can find each MFI or MII in time $\oof{\card{\idomo}}$, and the bound follows.

\begin{algorithm}[t]
   \SetAlgoLined
   \KwData {$\ont$: a taxonomy}
   \KwResult {$M=\mfi$ and $N=\mii$, for the correct $\freq$ predicate over $\onti$}
   $M,N\leftarrow \emptyset$\;
   \While {\emph{there is an \textsf{unclassified} element} $A\in\onti$} {
        \If{$\freqf[A]$} {
             \tcc{$A$ is an ancestor of an MFI, search for it by traversing $A$'s frequent descendants.}
            \For {$i\in \idomo$} {
                $B\leftarrow \funct{get-AC}{A\cup \funct{anc}{i}}$\\
                \lIf{$A < B$ \emph{\textbf{and}} $\freqf[B]$}{ $A\leftarrow B$}
            }
            \tcc{$A$'s descendants are infrequent}
            $\funct{mark-freq}{A}$; add $A$ to $M$\;
        }
        \Else{
            \tcc{$A$ is a descendant of an MII, search for it by traversing $A$'s infrequent ancestors.}
            \For {$i\in \idomo$} {
                $B\leftarrow \funct{get-AC}{A\backslash \funct{desc}{i}}$\\
                \lIf{$B < A$ \emph{\textbf{and}} $\neg\freqf[B]$}{ $A\leftarrow B$}
            }
            \tcc{$A$'s ancestors are frequent}
            $\funct{mark-infreq}{A}$; add $A$ to $N$\;
         }
   }
    \Return $M$, $N$;
    \caption{Identify $\mfi$ and $\mii$}
\label{alg:mfi_upper}
\end{algorithm}

\begin{theorem}
\label{thm:mfi_upper}
Algorithm~\ref{alg:mfi_upper} identifies $\freq$ in crowd complexity $\oof{\card{\idomo}\cdot\left(\card{\mfi} + \card{\mii}\right)}$.
\end{theorem}
\begin{proof}
We explain the course of Algorithm~\ref{alg:mfi_upper}, prove that it is correct (i.e., identifies $\freq$ correctly), and analyze its crowd complexity.

Algorithm~\ref{alg:mfi_upper} uses a few sub-routines: $\funct{mark-freq}{A}$ (resp., $\funct{mark-infreq}{A}$) classifies the itemset $A$ and its ancestors (resp., descendants) as \textsf{frequent} (resp., \textsf{infrequent}). $\funct{get-AC}{A}$ removes from $A$ all the items that are implied by other items (i.e., all $i \in A$ such that $i < i'$ for some $i' \in A$) so that $\funct{get-AC}{A}$ returns an antichain representing $A$. $\funct{anc}{i}$ and $\funct{desc}{i}$ return, respectively, the ancestors and descendants of $i$ in $\ont$ (including $i$).

We argue that each iteration of the main \textbf{while} loop of Algorithm~\ref{alg:mfi_upper} identifies exactly one new MFI or MII. First, an \textsf{unclassified} node $A\in\onti$ is chosen.
If $A$ is frequent (first \textbf{if} statement), it is either an MFI or an ancestor of an MFI. Since it used to be \textsf{unclassified}, at this point each of its descendants is \textsf{unclassified} or \textsf{infrequent}: in particular, $A$ is not an ancestor of an already discovered MFI. We thus start traversing
descendants of $A$ by adding items from $\idomo$ to $A$ and using $\name{get-AC}$ to turn the result into an
antichain.\footnote{We add $\funct{anc}{i}$ to $A$ to simplify the analysis in the next section; just adding $i$ would also work here.}
Either the current $A$ is an MFI so all of its children are \textsf{infrequent}, the inner \textbf{for} loop ends, and we identify $A$ as an MFI.
Otherwise, as $A$ is frequent but not maximal, there exists some frequent $B\in\onti$ s.t.\ $B=\funct{get-AC}{A\cup\funct{anc}{i'}}$ for some item $i'$. If $i'$ had already been considered by the \textbf{for} loop but was dismissed, it would mean that we dismissed an ancestor of $B$ as \textsf{infrequent}, contradicting the assumption that $B$ is frequent. Thus, $i'$ cannot have been considered by the \textbf{for} loop yet, so we will replace $A$ by $B$ before the \textbf{for} loop terminates. Hence, at the end of the \textbf{for} loop, we identify a new MFI.
In the same manner, the code within the \textbf{else} part
identifies an MII by traversing infrequent ancestors until reaching an \textsf{infrequent} element that has only \textsf{frequent} parents.

\emph{Correctness.}
The above implies that the algorithm terminates, that each identified MFI and MII is
correct, and that all elements are correctly marked as \textsf{frequent} and \textsf{infrequent}.
To prove completeness, consider an MFI $A$. By the end of the algorithm, $A$ is known to be \textsf{frequent}; since it has no \textsf{frequent} descendants,
$\funct{mark-freq}{A}$ was necessarily called, which implies that $A$ was added to $M$. The proof for MIIs is similar.

\emph{Complexity.} Since Algorithm~\ref{alg:mfi_upper} identifies an MFI or MII in each \textbf{while} iteration, there can be at most $\card{\mfi} + \card{\mii}$ iterations. The inner loop performs $\oof{\card{\idomo}}$ queries, and thus the total complexity is as stated above.
\end{proof}

Following an idea of~\cite{gunopulos2003discovering}, we observe that the bound
can be improved to $\oof{\card{\mii} + \card{\idomo} \cdot \card{\mfi}}$ if we
always choose the \textsf{unclassified} element $A$ to be minimal,
because this
ensures that no queries need to be performed whenever we are in the \textbf{else}
branch. Moreover, if we run two instances of Algorithm~\ref{alg:mfi_upper} in
parallel, one choosing maximal \textsf{unclassified} elements for $A$ and the
other one choosing minimal \textsf{unclassified} elements for $A$, we
improve the bound to
$\oof{\card{\mfi} + \card{\mii} +
  \card{\idomo} \cdot
\min\{\card{\mfi}, \card{mii}\}}$.

\subsection{Restricted Itemset Size}
\label{sec:limited}
We next consider the $k$-itemset taxonomy, $\ontik$.
Beyond the practical motivations for using $\ontik$ (see Section~\ref{sec:bg}), restricting the number of MFIs and MIIs may naturally improve the complexity bounds.

As explained in Section~\ref{sec:bg}, $\ontik$ is not necessarily a distributive lattice;
and the size of $\ontik$ is always polynomial while that of $\onti$ may be exponential (w.r.t.\ $\card{\idomo}$).
However, for every $\onti$ such that $k\geq \wid[\onti]$, it holds that $\ontik = \onti$.

Note that in Section~\ref{sec:query-wc} we did not make any assumptions on the itemset taxonomy structure, so our results apply to any poset and in particular to $\ontik$. We obtain the following, where $\ontsk \colonequals \ifunct{\ontik}$.

\Needspace{4\baselineskip}
\begin{corollary}
The worst-case crowd complexity of identifying $\freq$ over $\ontik$ is
$\omegaofi{\log\card{\ontsk}}$; and there exists an algorithm
to identify $\freq$ over $\ontik$ in crowd complexity $\oofi{\log\card{\ontsk}}\leq\oofi{\card{\idomo}^k}$.
\end{corollary}

For the complexity w.r.t.\ the output over restricted itemsets, the lower bound of Thm.~\ref{prop:mfi_lower} holds as well, using the same proof. For the upper bound, however, we cannot use Algorithm~\ref{alg:mfi_upper}: for a $k$-itemset taxonomy, adding (or removing) a single item to a $k$-itemset does not necessarily yield a $k$-itemset. Improving the trivial upper bound remains an open problem.

%% file: computational-complexity.tex
\section{Computational Complexity}
\label{sec:computational}  We next study the feasibility of
``crowd-efficient'' algorithms, by considering the computational
complexity of algorithms that achieve the upper crowd complexity
bound. We follow the same axes as in the previous section. In all
problem variants, we have the crowd complexity lower bound as a
simple (and possibly not tight) lower bound.
For some variants, we show that, even when the crowd complexity is feasible, the underlying computational complexity may still be infeasible.

\subsection{With Respect to the Input}
\label{sec:comp-input}
As a simple lower bound, we know that
the computational complexity of \aprob{} is higher than the crowd complexity,
and is thus $\omegaof{\log\left(\card{\onts}\right)}$.

The problem of finding tighter bounds for computational complexity
w.r.t.\ the input remains open.
Many works~\cite{dubhashi1993searching,faigle1986searching} provide efficient algorithms for computing a good
split element in particular types of posets, but no efficient algorithm is
known for the more general case of distributive lattices (or for
arbitrary posets). We now give evidence suggesting that no such
algorithm exists.

At any point of a \aprob{}-solving
algorithm, we define the \emph{best-split} element as the element of $\onti$ which is
guaranteed to eliminate the largest number of solutions of $\onts$
when queried. Following the proof of Prop.~\ref{prop:upper-crowd-input}, if
we could efficiently compute the best-split element, we would obtain
a computationally efficient greedy algorithm that is also
crowd-efficient. We show below that this is not possible for the
case of bounded-size itemsets (and their corresponding restricted
itemset taxonomies). This, of course, does not prove that there
exists no computationally efficient non-greedy algorithm;
however, it suggests that it is unlikely that such an algorithm
exists, because of the close relationship between finding best-split
elements and counting the antichains of $\onti$.
This is similar to a result of~\cite{faigle1986searching}, which proves that identifying a good-split element (which guarantees eliminating a constant fraction of the solutions) is computationally equivalent to a relative approximation of the number of order
ideals (though this is not known to be $\mathrm{\#P}$-complete).

\begin{theorem}
\label{thm:bestsplit}
~~The problem of identifying, given $\ont$ and $k$, the best-split
element in $\ontik$ is
$\fpsp$\!-complete\footnote{$\mathrm{\#P}$ is the
class of
\emph{counting} problems that return the number of solutions of
$\mathrm{NP}$ problems. $\fpsp$ is the class of
\emph{function} problems that can be
computed in polynomial time using a $\mathrm{\#P}$ oracle.} in
$\card{\ont}$.
\end{theorem}

\begin{proof} (Sketch).
  To prove membership, we show a
reduction from our problem to counting antichains in
a general poset, which is known to be in
$\mathrm{\#P}$~\cite{provan1983complexity}. Using an oracle
for antichain counting, we can count the number of eliminated
antichains in $\ontsk$ for every element of $\onti$, and thus find
the best-split element.

The more challenging part is proving hardness. For that, we show a reduction from the problem of counting antichains (which is $\fpsp$-hard) to our problem.
Let us call \emph{ancestor} and \emph{descendant} solutions of $A$ the solutions (elements of $\ontsk$) that are eliminated if
an itemset $A$ is discovered to be frequent or infrequent respectively.
For any poset $P$ and natural number $n$, we show that we can
construct
a $k$-itemset taxonomy $\ontik$ with an itemset $A_0$ such that, for some
increasing affine function $\name{F}$, $A_0$ has $\funct{F}{\card{\prop{P}}}$ descendant solutions and $\funct{F}{n}$ ancestor solutions.
As the best-split element $A^*$ in $\ontik$ has a roughly equal number of ancestor and descendant solutions,
comparing the position of $A_0$ and $A^*$
allows us to compare
$\card{\prop{P}}$ and $n$: if $A^*$ is an ancestor of $A_0$, it has more descendant solutions than $A_0$, and hence $\card{\prop{P}}<n$. Similarly, if $A^*$ is a descendant of $A_0$, $\card{\prop{P}}>n$.
Using this decision method, it is possible to perform a binary search on values of $n$ between~0 and $2^{\card{P}}$ and find the exact value of $\card{\prop{P}}$.
\end{proof}

As for upper bounds, our results for complexity w.r.t.\ the input, namely Cor.~\ref{cor:comp_crowd_upper}, will follow from the results w.r.t.\ the input and output that we present in the next section.

\needspace{5\baselineskip}
\subsection{With Respect to the Input and Output}
\label{sec:comp-output}

\vspace{-1em}
\paragraph*{Lower Bound}
As shown by Algorithm~\ref{alg:mfi_upper}, finding an MFI or MII
requires a number of queries linear in $\card{\idomo}$. However, note that the
algorithm assumes that at any point we are able to determine if the
set of unclassified elements of the itemset taxonomy is empty. We next show that this is a non-trivial problem.
We recall the definition of problem EQ~\cite{bioch1995complexity}.
Let $B_n = \{0, 1\}^n$ be the set of Boolean vectors of length~$n$.
Define the order $\leq$ on $B_n$ by $x \leq y$ iff $x_i \leq y_i$ for all~$i$.
For $X \subseteq B_n$, write $T(X) = \{y \in B_n \mid \exists z \in X,\,z \leq y\}$ and $F(X) = \{y \in B_n \mid \exists z \in X,\,y \leq z\}$.
Problem EQ is the following: given $X, Y \subseteq B_n$ such that
$T(X) \cap F(Y) = \emptyset$, decide whether $T(X) \cup F(Y) = B_n$.

\begin{proposition}
\label{prop:comp_output_lower}
If \aprob{} can be solved in computational time
$\oof{\funct{poly}{\card{\mii},\card{\mfi},\wid}}$ then there exists a PTIME solution for problem EQ from~\cite{bioch1995complexity}.
\end{proposition}
It is unknown whether EQ is solvable in polynomial time
(see~\cite{eiter2008computational,gottlob2013deciding} for a survey);
the connection between
frequent itemset mining
and EQ (and its other equivalent formulations, such as monotone dualization or hypergraph transversals)
was already noted in~\cite{mannila1997levelwise}.
Note that the proof above
uses the fact that the itemset size is not restricted. For
$k$-itemset taxonomies, finding a tighter lower bound than the trivial $\card{\mfi}+\card{\mii}$ remains an
open problem.

\paragraph*{Upper Bound}
We consider again Algorithm~\ref{alg:mfi_upper}, whose crowd complexity we analyzed in
Section~\ref{sec:crowd_output}. By completing some
implementation details, we can now analyze its computational complexity as
well, and obtain an upper bound. For simplicity, this bound is presented in the Introduction with
$\card{\ont}$ which is $\geq\card{\idomo}$.

\begin{proposition}
\label{prop:upper_comp_output}
There exists an algorithm
to solve \aprob{} in computational time \[\oofi{\card{\onti}\cdot(\card{\idomo}^2+\card{\mfi}+\card{\mii})}\]
\end{proposition}
\begin{proof}
Algorithm~\ref{alg:mfi_upper} uses a computation of an unclassified
element of $\onti$. Since by Prop.~\ref{prop:comp_output_lower} this
is probably non-polynomial, we can use the brute-force method of materializing
the itemset taxonomy $\onti$.
We use a hash table to find any element in the
$\onti$ structure in time linear in the element size. The implementation of $\name{mark-freq}$ and
$\name{mark-infreq}$ locates $A$ in $\onti$ using the hash
table, traverses its ancestors or descendants respectively, and marks
them as (in)frequent.

To compare itemsets efficiently, we represent each
itemset $A$ by an ordered list of the items in its order ideal,
i.e., $\ideal{A}=\left\{i\in\idomo \mid \exists i'\in A, i'\orelit
i\right\}$. In this case, $A\orelit B$ iff $\ideal{A}\subseteq \ideal{B}$,
which can be verified in time
$\oofi{\card{\ideal{A}}+\card{\ideal{B}}}\leq\oofi{\card{\idomo}}$. Using this
representation, we do not need
the sub-routine $\name{get-AC}$. We generate once, for every $i\in\idomo$, two
ordered lists: $\funct{desc}{i}$ and $\funct{anc}{i}$, holding its
descendants and ancestors respectively. These lists can be computed
in time $\oofi{\card{\idomo}^2}$ by building the transitive closure
of $\onti$,
and can be used to compute $\ideal{A} \cup
\funct{anc}{i}$ and $\ideal{A} \backslash \funct{desc}{i}$ in time $\oofi{\card{\idomo}}$.

Let us analyze the overall complexity of the suggested
implementation.
We construct $\onti$ (where each element has both its antichain and order ideal representations)
in $\oofi{\card{\idomop}\cdot\card{\idomo}^2}$ according to
Prop.~\ref{prop:contruct} (deferred to the Appendix), and construct $\funct{anc}{i}$ and
$\funct{desc}{i}$ in time $\oofi{\card{\idomo}^2}$.
Now, we run $\card{\mfi}+\card{\mii}$ times
the body of the outer \textbf{while} loop, which
\begin{inparaenum}
\item finds an unclassified element by a brute-force search taking time $\card{\idomop}$,
\item runs $\oofi{\card{\idomo}}$ times the body of one of the \textbf{for}
  loops that computes $\ideal{A}
  \cup \funct{anc}{i}$ or $\ideal{A}\backslash \funct{desc}{i}$ and verifies
$\orelit$ in time $\oofi{\card{\idomo}}$, and
\item calls $\name{mark-freq}$ or $\name{mark-infreq}$ which takes
time $\oofi{\card{\idomo}+\card{\onti}}$ to locate the itemset in
$\onti$ and traverse its ancestors or descendants.
\end{inparaenum}
Summing these numbers and simplifying the expression
yields the claimed complexity bound.
\end{proof}

Since we know that $\card{\mfi}+\card{\mii}\leq\card{\idomop}$, we
can plug $\card{\idomop}$ in the complexity formula and obtain an
upper bound that does not depend on the numbers of MFIs and
MIIs. In this manner we achieve a bound polynomial in $\card{\onti}$
and improve the upper bound described in
Section~\ref{sec:comp-input}. However, note that this is in fact a relaxation of our requirement for crowd-efficient algorithms, since Algorithm~\ref{alg:mfi_upper} is not crowd-efficient w.r.t.\ the upper bound of Prop.~\ref{prop:upper-crowd-input}, in terms of the input.
This result is also simplified in the Introduction, replacing $\card{\idomo}$ by $\card{\ont}$ which is $\geq\card{\idomo}$, and $\card{\idomop}$ by $\card{\onti}$ which is $\geq\card{\idomop}$.

\begin{corollary}
\label{cor:comp_crowd_upper}
There exists an algorithm
to solve \aprob{} in computational
complexity
\[\oofi{\card{\onti}\cdot(\card{\idomo}^2+\card{\idomop})}\]
\end{corollary}

%% file: heuristics.tex
\section{Chain Partitioning}
\label{sec:relax}
Recall that in the beginning of Section~\ref{sec:query-wc} we mentioned the special case of chain taxonomies, for which
a binary search achieves a tight complexity bound, both crowd and computational, of $\thetaof{\log\card{\idomo}}$.
We generalize this insight to solve \aprob{} for taxonomies
partitioned in disjoint chain taxonomies.
\emph{Chain partitioning} is a standard technique in Boolean function learning~\cite{hansel1966nombre,kovalerchuk1996interactive}, that splits the Boolean lattice elements into disjoint chains, and then performs a binary search for the maximal frequent element on each chain.
The following
easy proposition holds (we justify how the partition $P$ is obtained at the end of the section):

\begin{proposition}
  \label{prop:chain_complex}
Given a partition $P$ of $\onti$ into $\wid[\onti]$ chains, $\freq$ can be identified in both crowd and computational complexity $\oof{\wid[\onti]\cdot\log\card{\idomo}}$.
\end{proposition}

The $\log\card{\idomo}$ factor comes from the binary search in the chains. To understand intuitively why their length is at most $\card{\idomo}$, notice that the worst case is achieved by the full Boolean lattice, and that, in this case, for every chain of the form $A_0\orelit \dots \orelit A_n$, it holds that $\card{A_i}+1\leq\card{A_{i+1}}$, so at most $\card{\idomo}$ items can be added to $A_0$ in total (see Figure~\ref{subfig:flati}).

Let us compare the result of Prop.~\ref{prop:chain_complex} with previous results. In terms of crowd complexity, if $\card{\onts}$ is close to its lower bound, $2^{\wid[\onti]}$, then the partition binary search performs more queries by a multiplicative factor of $\log\card{\idomo}$ than the upper bound of Prop.~\ref{prop:upper-crowd-input}. On the other hand, since we know that the bound of Prop.~\ref{prop:upper-crowd-input} is tight, we get an upper bound for $\card{\onts}$ that depends on $\wid[\onti]$ (in addition to the trivial upper bound $2^{\card{\onti}}$).

\begin{corollary}
$\card{\onts}\leq 2^{\wid[\onti]\log\card{\idomo}}$.
\end{corollary}

When $\card{\mfi}+\card{\mii} = \omegaof{\wid[\onti]}$, the crowd complexity of the partition binary search is asymptotically smaller than that of Algorithm~\ref{alg:mfi_upper}, $\oof{\card{\idomo}\cdot\left(\card{\mfi} + \card{\mii}\right)}$. Intuitively, this is because Algorithm~\ref{alg:mfi_upper}, in the worst case,
can traverse a full chain for every MFI and
MII, taking linear time whereas the partition binary search takes logarithmic time. However, when $\card{\mfi}+\card{\mii}$ is small w.r.t.\ $\wid[\onti]$, Algorithm~\ref{alg:mfi_upper} considers significantly less chains and is thus more efficient.

It remains to explain how to obtain the partition $P$. By Dilworth's theorem, it is possible to partition the poset $\onti$ into exactly $\wid[\onti]$ chains~\cite{dilworth1950decomposition}.
Computing the partition can be done in $\oof{\funct{poly}{\card{\onti}}}$, by a reduction to maximum matching (or maximal join) in a bipartite graph~\cite{fulkerson1956note}.
See Appendix~\ref{sec:partition_app} for a discussion on the complexity of taxonomy chain partitioning.

%% file: greedy.tex
\section{Greedy Algorithms}
\label{sec:greedy}
In the previous sections, we have attempted to fully identify
$\freq$. The solutions that we presented try to do so by maximizing the number
of eliminated solutions, or identifying MFIs or MIIs. However, we may not be able to pose enough questions to identify
$\freq$ exactly. In a dynamic crowd setting we could assume, e.g., that
the cost of obtaining answers from the crowd (both in terms of money and
latency) is not
controlled, and that the identification of $\freq$ may be
interrupted at any time. In such cases, our algorithms would perform
badly:

\begin{example}
Assume that the \emph{unclassified} part  of
the itemset taxonomy $\onti$ contains a chain $C$ of even length
$2n$, for some $n>1$, and one incomparable itemset $A$. There are
exactly $2+4n$ antichains in this poset (one empty, $1+2n$ of size~1
and~$2n$ of size~2), which is also the number of possible solutions.
Asking about $A$
eliminates exactly half of the possible
solutions for $\freq$
and finds an MII or MFI. However, if we
have to interrupt the computation after only one query,
we have only obtained
information about $A$. It would have been
better
to query a middle element of $C$: though this eliminates less solutions and
does not identify an MII of MFI, it classifies $\geq n$ itemsets.
\end{example}

Motivated by this example, in this section, we assume that the computation can
be halted at any time, and look at the intuitive strategy that tries to maximize
the number of classified itemsets at halting time using the following greedy
approach: compute, for every itemset, what is the worst-case (minimal) number of itemsets
that could be classified if we query it; then query the \emph{greedy best-split
itemset}, namely the itemset which \emph{maximizes} this number. To perform the
greedy best-split computation, we need to count the number of ancestors and
descendants of each element; this may be done in time linear in $\card{\onti}$
per itemset. In terms of $\card{\ont}$, we can show that this computation is
hard.

\begin{proposition}
\label{prop:greedy} Finding the greedy best-split itemset is
$\mathrm{FP}^{\mathrm{\#P}}$-hard  w.r.t.\ $\card{\ont}$. There
exists an algorithm which finds it in time
$\oofi{\card{\idomop}\cdot(\card{\idomo}^2+\card{\onti})}$.
\end{proposition}

To prove this, we first observe that the structure of
$\ontsk[1]$  is almost identical to that of $\onti$. Then, for the
structure used in the proof of Thm.~\ref{thm:bestsplit}, we show a
reduction from finding the best-split element in $\ontsk[1]$ to
finding the greedy best-split element in $\ontik[1]$. See
Appendix~\ref{sec:greedy_app} for the
details. The second part follows from the brute-force method described above, in combination with the complexity of materializing $\onti$ (see the upper bound in Appendix~\ref{sec:struct}).

%% file: related.tex
\Needspace{4\baselineskip}
\section{Related Work}
\label{sec:related}
Throughout the paper, we have combined and extended results from \emph{order
theory}, \emph{Boolean function
learning} and \emph{data mining}~\cite{bioch1995complexity,bshouty1995exact,gainanov1984one,gunopulos2003discovering,gunopulos1997data,linial1985every},
to obtain our characterization of the complexity of the crowd mining problem.
We now discuss further related work.

Several recent works consider the use of crowdsourcing platforms as
a powerful means of data
procurement (e.g.,~\cite{franklin2011crowddb,marcus2011crowdsourced,parameswaran2012deco}).
As the crowd is an expensive resource, many works focus on
minimizing the number of questions posed to the crowd to perform a
certain task: for instance,
computing common query operators such as filter,
join and
max~\cite{davidson2013crowd,guo2012dynamic,marcus2011human,parameswaran2013crowdscreen,venetis2012max},
performing
entity resolution \cite{wang2012crowder}, etc.
The present work considers the mining of data patterns from the crowd, and thus is closely related
to this line of work.

The most relevant work, by some of the present authors,
is~\cite{amsterdamer2013crowd}, which proposes a general crowd mining
framework. That work focused on a technique to estimate the
confidence in a mined data pattern and how much it increases if
more answers are gathered: we could use this technique to implement the crowd
query black-box mechanism in our context. However,~\cite{amsterdamer2013crowd} did not address the issue
of the dependencies between rules, or study the implied complexity boundaries, which
is the objective of the present paper.
Another
particularly relevant work is~\cite{parameswaran2011human}, which
considers a crowd-assisted search problem in a graph.
While it is possible to reformulate
some of our problems as graph searches in the itemset and solution
taxonomies, there are two important differences between
our setting and theirs. First, our itemset and solution taxonomies may be exponential in
the size of the original taxonomy but have a specific structure, which allows, in some cases, to perform the search without materializing them. Second,
we allow algorithms for \aprob{} to choose crowd queries interactively based on
the answers to previous queries, whereas~\cite{parameswaran2011human} studies
``offline'' algorithms where all questions are selected in advance. Consequently, our algorithms and complexity results are inherently different.

Frequent itemset discovery is a fundamental building
block in data mining algorithms (see, e.g.,~\cite{agarwal1994fast}).
The idea of using taxonomies in data mining was suggested in~\cite{srikant1995mining}, which we use as a
basis for our definitions.

Another line of works in data mining models the discovery of interesting data patterns through oracle calls~\cite{mannila1997levelwise}.
This work is closely connected to ours
by (i) the use of oracles, which may be seen as an abstraction of the crowd (compared to our setting), and
(ii) the separation between the complexity analysis of the number of oracle calls (crowd complexity in our case) and of the computational process.
However, because our motivation is to query the crowd, we focus on the specific problem of mining under a taxonomy over the itemsets (and related variants such as limiting the itemset size) which is not studied in itself in this line of work.
On the one hand, \cite{mannila1997levelwise} studies a generalization of our setting, namely the problem
of finding all interesting sentences given a specialization relation on
sentences. They introduce the notion of border (corresponding to MFIs and MIIs) as a way to bound the number of oracle calls.
However, in this general setting, they are not able to give complexity bounds on the performance of applicable algorithms (e.g., Algorithm All\_MSS from \cite{gunopulos1997discovering}) to match the bounds that we obtain for the more specific setting of mining frequent itemsets under a taxonomy.
On the other hand, the aforementioned papers also study the restricted case of Boolean lattices
and give complexity bounds in this case (e.g., for the Dualize and Advance
algorithm~\cite{gunopulos2003discovering,gunopulos1997data}); however, those algorithms
exploit the connection with hypergraph
traversals which is very specific to the Boolean lattice.
Hence, these algorithms cannot be used to mine frequent itemsets under a taxonomy, which is very natural when working with the crowd, and their complexity bounds are not applicable to our problem.
Finally, among the many works that discuss the connection of data mining and hypergraph traversals, we note the recent
work~\cite{gottlob2013deciding} which is relevant to our EQ-hardness result
(Prop.~\ref{prop:comp_output_lower}) as it sheds more light on the (still open) complexity of EQ.

%% file: conc.tex
\section{Conclusion}
\label{sec:conc}
In this paper, we have considered the identification of frequent itemsets in human knowledge
domains by posing questions to the crowd, under a taxonomy which captures
the semantic dependencies between items. We studied the complexity boundaries of solutions to this problem, in terms of two cost metrics: the number of crowd queries required for identifying the frequent itemsets, and the computational complexity of choosing these queries. We identified two main factors that affect both complexities: the structure of the taxonomy; and properties of the frequency predicate.

Our results leave some intriguing theoretical questions open: in particular,
we would like to
find tighter complexity bounds where possible, and to
further study the nature of the tradeoff between crowd and computational complexities. In addition, due to the high complexity of taxonomy-based crowd mining, practical implementations
could further resort to approximations and randomized algorithms in order to identify (in expectation) a large portion of the frequent itemsets, while reducing the complexity. The greedy approach mentioned in Section~\ref{sec:greedy} forms a first step in this direction of further research. A different approach involves filtering the itemsets according to a user request, which could
reduce the solution search space: for instance, the user may wish to mine itemsets composed of small fragments of the taxonomy, or
respecting certain constraints. We intend to investigate this approach in future work.

\paragraph*{Acknowledgments}
We are grateful to the anonymous referees and Toon Calders for their useful comments that have helped us improve our paper, and in particular for pointing us to related work.
We also thank Pierre Senellart for his insightful ideas and comments.

%% file: appendix.tex
\appendix
\section{Itemset Taxonomy Structure}
\label{sec:struct}

\vspace{-1em}

\paragraph*{Itemset taxonomies and distributive lattices}
We noted in Section~\ref{sec:bg} that itemset taxonomies are not arbitrary posets. In fact, by the observation that antichains correspond to order ideals (or lower sets), from the proof of Prop.~\ref{prop:lower-query}, their structure is characterized by Birkhoff's representation theorem.\footnote{Birkhoff, G. (1937). ``Rings of sets''. \emph{Duke Mathematical Journal}, 3(3), 443--454.} A \emph{distributive lattice} is a standard mathematical structure where join $\wedge$ and meet $\vee$ operations are defined (which roughly correspond to AND and OR) and where these operations
distribute over each other ($(x\wedge y)\vee z = (x\vee z)\wedge(y\vee z)$ and likewise if we exchange $\vee$ and $\wedge$).

\begin{theorem}[Birkhoff '37]
Any itemset taxonomy is isomorphic to a \emph{distributive lattice}, and vice versa.
\end{theorem}

\paragraph*{The covering relation}
When representing $\onti$ as a DAG (like in Figure~\ref{subfig:onti1}), with edges representing the covering relation $\oreli$, it is not easy to see which relationship holds between the itemsets at each end of an edge. However, looking more closely, we can characterize those edges as two different types:

\begin{compactitem}
\item \textbf{Addition edge:} An edge between two itemsets $A$, $B$ where there exists $i\in\idom{}$ s.t.\ $B=A\cup\{i\}$, i.e., one item $i$ is added to $A$ to obtain $B$.
(The item $i$ is necessarily one of the maximal items that are not implied by those of $A$, i.e., necessarily every $j$ such that $j \oreli i$ is implied by $A$.)
\item \textbf{Specialization edge:} An edge between two itemsets $A$, $B$ where there exists $i,i'\in\idom{}$ s.t.\ $i\orel{}i'$ and $B=A-\functsq{parents}{i'}\cup\{i'\}$, i.e., at least one item in $A$ is made more specific to obtain $B$.
\end{compactitem}

The reason for these two edge types becomes clear if we represent itemsets by their order ideals.
Denote by $\ideal{A}$ and $\ideal{B}$ the order ideals of itemsets $A$ and $B$ respectively. The following easy proposition holds.

\begin{proposition}
\label{prop:add_spec}
$A\oreli B$ in $\onti$ iff there exists $i\not\in A$ s.t.\ $\ideal{B}=\ideal{A}\cup\{i\}$ and $\functsq{parents}{i}\subseteq \ideal{A}$
\end{proposition}
\begin{proof}
First observe that $A\orelit B$ iff $\ideal{A}\subseteq \ideal{B}$ (this follows from the correspondence between taxonomies and distributive lattices).

For one direction, assume $A \oreli B$. Then $B \not\orelit A$, so $\ideal{B} \not\subseteq \ideal{A}$. Let $i$ be an element of $\ideal{B} \backslash\!\!\!\ideal{A}$ that is minimal for the $\orelt$ order on $\idomo$. Clearly $A \orelit A \cup \{i\}$, and $\ideal{A \cup \{i\}} \subseteq \ideal{B}$ so $A \cup \{i\} \orelit B$. Because $i \notin \ideal{A}$, $A \cup \{i\} \neq A$, so $A \orelit A \cup \{i\} \orelit B$ and $A \oreli B$ implies that $B = A \cup \{i\}$. To show that $\functsq{parents}{i}\subseteq \ideal{A}$, observe that for any $j \in \functsq{parents}{i} \backslash\!\!\!\ideal{A}$ we would have $j \in \ideal{B} \backslash\!\!\!\ideal{A}$, contradicting the minimality of $i$.

Conversely, let $A$ and $B$ be itemsets such that there exists $i\not\in A$ s.t.\ $\ideal{B}=\ideal{A}\cup\{i\}$ and $\functsq{parents}{i}\subseteq \ideal{A}$. Because $\ideal{A} \subseteq \ideal{B}$, $A \orelit B$. Consider $C$ such that $A \orelit C \orelit B$. Because $\ideal{A} \subseteq \ideal{C} \subseteq \ideal{B}$, and because the condition on $i$ imposes that $\card{\ideal{B}} = \card{\ideal{A}} + 1$, we must have $\ideal{C} = \ideal{A}$ or $\ideal{C} = \ideal{B}$, so $C = A$ or $C = B$, thus $A \oreli B$.
\end{proof}

By the proposition above, there are two cases in $\onti$ in which $A\oreli B$: let $i$ be the item s.t.\ $\ideal{B}=\ideal{A}\cup\{i\}$.
If all the parents of $i$ are in $\ideal{A}$ but not in $A$ (including if $i$ has no parents), then necessarily $i$ has no ancestors in $A$ (by maximality of the elements of $A$ in $\ideal{A}$) and an \emph{addition} edge will connect $A$ and $B$.
Otherwise, the parents of $i$ in $A$ must be removed, or \emph{specialized}, in order to obtain $B$.

\begin{example}
In Figure~\ref{fig:ex_lemcomp} the solid arrows stand for addition edges, and the double arrows stand for specialization edges. (We ``overload'' this representation and use double arrows in the item taxonomies $P$ and $\ont$, in order to denote the semantic relationship between items.)
For instance, in $\ontsk[1]$ there is an addition edge between $\{\{5\}\}$ and $\{\{5\},\{6\}\}$ since the parent of $\{6\}$ in $\ontik[1]$, $\{\}$, is implied by $\{5\}$ and thus $\{6\}$ can be added; and there is a specialization edge from $\{\{5\},\{6\}\}$ to $\{\{7\}\}$ since both $\{5\}$ and $\{6\}$ are specialized into their child in $\ontik[1]$, $\{7\}$.
\end{example}

\paragraph*{Materializing the itemset taxonomy}
We next describe an explicit process to materialize $\onti$ and $\ontik$, which is used
in Section~\ref{sec:computational} of this paper. There are naturally different possible representations for both itemsets (e.g., the itemset and the corresponding order ideal) and orders (Hasse diagram, full transitive closure of the relation...). We choose a representation that allows for an efficient itemset taxonomy construction, and which can be later translated into other representations.

We assume some total order on all items of $\ont$, respecting $\orelt$, which may be, e.g., a topological ordering of $\ont$. For each element $A$ of the itemset taxonomy we keep~2 ordered lists: $E_A$ which contains the items in $A$ and $O_A$ which contains all the elements in the order ideal $\ideal{A}$ corresponding to $A$, or in other words, all the ancestors of the items in $E_A$ including the items of $E_A$ themselves. We also assume that given an element in $\ont$, accessing one of its children or parents can be done in $\oof{1}$.

\begin{algorithm}[t]
   \SetAlgoLined
   \KwData {$\ont$: a taxonomy}
   \KwResult {$\onti$: the itemset taxonomy of $\ont$}
   $E_\emptyset, O_\emptyset\leftarrow \emptyset$\;
   \tcc{$Q$ contains the itemsets to handle.}
   $Q\leftarrow$ a FIFO queue with only $\emptyset$\;
   \tcc{$D$ contains the handled itemsets.}
   $D\leftarrow$ an empty hash table\;
   $O \leftarrow$ an empty set of itemset pairs\;
   \While {$Q$ is not empty} {
        $A\leftarrow \functsq{pop}{Q}$\;
        Add $A$ to $D$\;
        \For{$i\in \idom$} {
          \If{$i \notin O_A$ \emph{\textbf{and}} $\functsq{parents}{i}\subseteq O_A$} {
                Create a new itemset $B$\;
                $E_B\leftarrow \left(E_A-\functsq{parents}{i}\right)\cup \{i\}$\;
                $O_B\leftarrow O_A\cup \{i\}$\;
                Add $\langle A, B \rangle$ to $O$\;
                \If{$B\not \in D$} {
                    Add $B$ to $Q$\;
                }
            }
        }
    }
    \Return $\onti =(D, O)$;
\caption{Construct $\onti$}
\label{alg:construct}
\end{algorithm}

Algorithm~\ref{alg:construct} constructs, given $\ont$, the itemset taxonomy $\onti$. The following proposition proves its correctness and complexity.

\begin{proposition}
\label{prop:contruct}
Given a taxonomy
$\ont$ as input, Algorithm~\ref{alg:construct} constructs $\onti$ in time $\oofi{\card{\idomop}\cdot\card{\idomo}^2}$.
\end{proposition}
\begin{proof}
\emph{Completeness.}
We show the following induction: whenever $A$ is constructed properly, $A \oreli B$, and $A \in D$, then $B$ is constructed properly, $B \in D$ and $\langle A, B \rangle \in O$. The base case of $A = \emptyset$ is straightforward.

Consider the iteration at which $A$ is handled in the \textbf{while} loop.
By Prop.~\ref{prop:add_spec}, there exists $i\in B$ s.t.\ $\ideal{B}=\ideal{A}\cup\{i\}$ and $\functsq{parents}{i}\subseteq \ideal{A}$. Because $A$ was properly constructed, we have $i \notin O_A$ and $\functsq{parents}{i}\subseteq O_A$ and therefore we enter the \textbf{then} block. The construction of $O_B$ is clearly correct. As $E_A$ was correctly constructed, the only ancestors of $i$ in $E_A$ can be direct parents of $i$, which are removed. No descendants of $i$ can be in $E_A$ by definition. Thus, by removing the parent of $i$ we obtain the antichain representing the order ideal $\ideal{B}$ i.e., the correct $E_B$. Finally, we add $\langle A, B \rangle$ to $O$, and if $B$ is not yet in $D$ it is added to $Q$ so will be added to $D$ later.

We therefore conclude by induction that $\prop{\ont} \subseteq D$ and $\oreli \subseteq O$.

\emph{Correctness.}
From the above it is immediate that $D \subseteq \prop{\ont}$. It remains to prove that we do not add incorrect pairs to $O$. Assume some pair of antichains $\langle A, B \rangle$ was added by the algorithm to $O$. This means that their order ideals differ by a single item satisfying the conditions of Prop.~\ref{prop:add_spec}; but then $\oreli$ must hold for the pair.

\emph{Complexity.}
The \textbf{while} loop traverses all the elements in $\onti$ - there are $\card{\prop{\ont}}$ of them. For each element the inner \textbf{for} loop traverses $\card{\idomo}$ items. Other actions within the loop -- verifying that $i \notin O_A$ and all of $i$'s parents are in $O_A$, constructing $B$, etc.\ -- take $\oof{\card{\idomo}}$ actions. Thus, the total complexity is as stated.
\end{proof}

Algorithm~\ref{alg:construct} is fairly simple and we do not claim it is optimal. That said, for some taxonomies, e.g., Boolean lattices, $\card{\onti} = \thetaof{\card{\idomo}\cdot \card{\idomop}}$. Hence, for such taxonomies our algorithm is at most within a multiplicative factor of $\card{\idomo}$ from the construction lower complexity bound.

Now, we turn to construct $\ontik$. We cannot use a straightforward adaptation of Algorithm~\ref{alg:construct} for this purpose, since it relies on some assumptions that are not valid
in for $k$-itemset taxonomies. For instance, if $A\orelik B$ in $\ontik$, then $A$ and $B$ are not necessarily separated by a single item. However, since $\ontik$ is small in size, we can propose the following straightforward algorithm:

\begin{proposition}
Given a taxonomy $\ont$, $\ontik$ can be materialized in time $\oofi{\card{\idomo}^{2k+1}}$.
\end{proposition}
\begin{proof}
If $k=0$, construct a single itemset, $\emptyset$.
Otherwise, find for every $i\in \idomo$ its set of ancestors, lexicographically ordered. This can be done in $\oofi{\card{\idomo}^2}$.
Generate all the subsets of $\idomo$ of size up to $k$ in time $\oofi{\card{\idomo}^k}$. For each such itemset $A$, check whether it is an antichain by checking for every $i\in A$, whether one of the ancestors of $i$ are in $A$. This takes $\oofi{\card{\idomo}^2}$ operations per itemset. Keep only the antichains -- these are the elements of $\ontik$. For each antichain also compute its order ideal, by computing the union of ancestors of $i\in A$, in time $\oofi{\card{\idomo}^2}$. Finally, for each pair of antichains $A, B$ check in $\oofi{\card{\idomo}}$ whether $A\orelikt B$ -- if the order ideal of $B$ contains the one of $A$. The total complexity is $\oofi{\card{\idomo}^2+\card{\idomo}^k\cdot\card{\idomo}^2+\card{\idomo}^{2k}}\cdot\card{\idomo}$, which, for $k\geq 1$, is $\oofi{\card{\idomo}^{2k+1}}$.
\end{proof}

\emph{Remark.} The above algorithm for $k\geq 1$ computes the full relation $\orelit$. If we are interested in $\oreli$,
we can perform the transitive reduction of the DAG represented by the pairs of $\orelit$. In particular, transitive reduction can be computed by performing, for every node, a linear-time longest path search from this node, and keeping only the paths of length~$1$. The total complexity is the number of nodes (in our case $\oofi{\card{\idomo}^k}$) times the number of edges, and thus the total complexity is $\oofi{\card{\idomo}^{3k}}$.

\section{Supplementary Proofs}
\label{sec:additional}
\subsection{Crowd Complexity}
\label{sec:crowd_app}
\begin{proof}[(Prop.~\ref{prop:onts_predicates})]
We prove the following more general result: for any poset $P$, there exists a bijective mapping $\phi$ between $\prop{P}$ and the set of monotone predicates over $P$. Our original claim will follow from $P = \onti$. Consider the mapping $\phi$ associating, to a monotone predicate $F$ over $P$, the set of its maximal elements. First, observe that $\phi$ is clearly injective. Next, observe that the set of maximal elements is an antichain of $P$, because if two such elements are comparable, it contradicts the maximality of one of them. Thus, the range of $\phi$ is actually $\prop{P}$. Conversely, from any antichain $A\in\prop{P}$ we can define a predicate $F$ which returns \true{} for $e$ iff $\exists e'\in A~ e\leq e'$. By this definition, $F$ is monotone and has $A$ as its set of maximal elements. Hence, $\phi$ is also surjective.
\end{proof}

\begin{proof} [(Prop.~\ref{prop:lower-query}, direct)]
With every query of an itemset, either some set of solutions in $\onts$ is marked as \textsf{impossible}, or
its complement is. Thus, in the worst-case, a query eliminates at most half of the
\textsf{possible} solutions. Hence, by induction, in order to eliminate all solutions except the correct one, every algorithm must pose a number of queries that is at least logarithmic in the solution space, i.e., $\omegaof{\log\card{\prop{\onti}}}$ or, equivalently, $\omegaof{\log\left(\card{\onts}\right)}$.
\end{proof}

\begin{proof}[(Prop.~\ref{prop:upper-crowd-input})]
While we search for the right element in $\onts$, we can mark the itemsets of $\onti$ as \textsf{frequent}, \textsf{infrequent} and \textsf{unclassified} by using the result of our queries and the monotonicity
property: an itemset is marked as \textsf{frequent} if it is an ancestor of an itemset that was queried and identified as frequent (remember that an itemset is an ancestor of itself), as \textsf{infrequent} if it is the descendant of an itemset identified as infrequent, and as \textsf{unclassified} in other cases.
At any point of the algorithm, define $U\subseteq \onti$ as the set of itemsets that are marked \textsf{unclassified}, and $R\subseteq \onti$ as the set of itemsets marked \textsf{frequent} that have
no descendants marked \textsf{frequent}.
We claim that there is an order-preserving bijection from the antichains of $U$ to the \textsf{possible} elements of $\onts$ . Consider the function $\phi$ associating, to an antichain $A$ of $U$, the subset $\phi(A)$ of $\onti$ formed by $A \cup B$, where $B$ is the set of elements in $R$ that have no descendant in $A$. $\phi(A)$ is an antichain of $\onti$ because $A$ is an antichain, $B$ is a subset of $R$ which is an antichain, and any element of $U$ is incomparable to any element of $R$. $\phi(A)$ is \textsf{possible} in $\onts$, because it is consistent with the observations. $\phi$ is injective because $A \mapsto \phi(A) \cap R$ is the identity. The range of $\phi$ is the \textsf{possible} elements of $\onts$, because any \textsf{possible} element of $\onts$ must be an antichain of $R \cup U$ including a descendant of every itemset of $R$. $\phi$ is order-preserving.
Thus, we can characterize the \textsf{possible} elements of $\onts$ using only $U$, and in fact, they form a poset isomorphic to the itemset taxonomy $\ifunct{U}$.

By~\cite{linial1985every}, in any poset there exists an element $e$ such that the fraction of order ideals that contain $e$ is between~$\delta_0\cong0.17$ and $1-\delta_0$. Thus, there exists an itemset $A\in U$ such that the fraction of \textsf{possible} solutions in $\onts$ that contain $A$ or one of its descendants is between $\delta_0$ and $1-\delta_0$. Consequently, when querying $A$ we are guaranteed to eliminate at least~$\delta_0$ of the \textsf{possible} elements in $\onts$.

We can define an algorithm which achieves this upper bound as follows: at each iteration, choose a ``good split'' element, e.g., by counting for each element in $\onti$ the fraction of possible solutions in $\onts$ which contain this element or its descendants.
This algorithm terminates after $\oofi{\log_{1/(1-\delta_0)}\card{\onts}}=\oofi{\log\card{\onts}}$ queries.
\end{proof}

\begin{proof}[(Prop.~\ref{prop:mfi_lower})]
Assume w.l.o.g.\ that we have not queried some MFI $A$. (The argument is similar if $A$ is an MII.) Because $A$ is frequent, we have $\freqf[A] = \true{}$. Denote by $\freq'$ the predicate
s.t.\ $\freq'(A) = \false$ and $\freq'(B) = \freqf[B]$ for any $B\neq A$. As we
have not queried $A$, we cannot distinguish between $\freq$ and
$\freq'$: since all of $A$'s children are infrequent, and all of its ancestors are frequent, determining $\freqf[A]$ can only be done by querying $A$ directly. Since $\mfi \cap \mii = \emptyset$, $\freq$ cannot be identified by less than $\card{\mfi} + \card{\mii}$ queries.
\end{proof}

\begin{figure*}
\centering
\begin{subfigure}[b]{0.15\textwidth}
        \centering
        \includegraphics[scale=0.35]{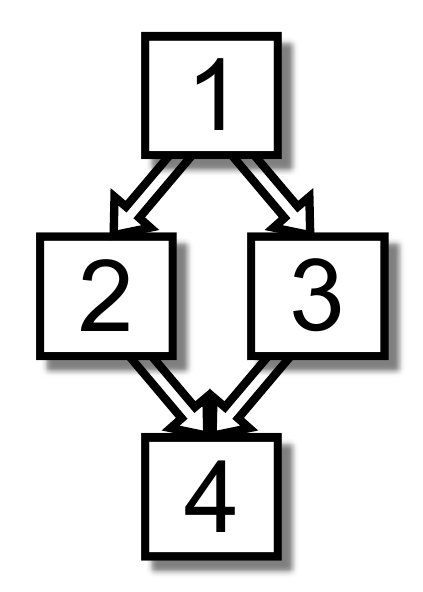}
        \caption{Original input $P$}
        \label{subfig:bin_p}
\end{subfigure}
\begin{subfigure}[b]{0.3\textwidth}
        \centering
        \includegraphics[scale=0.35]{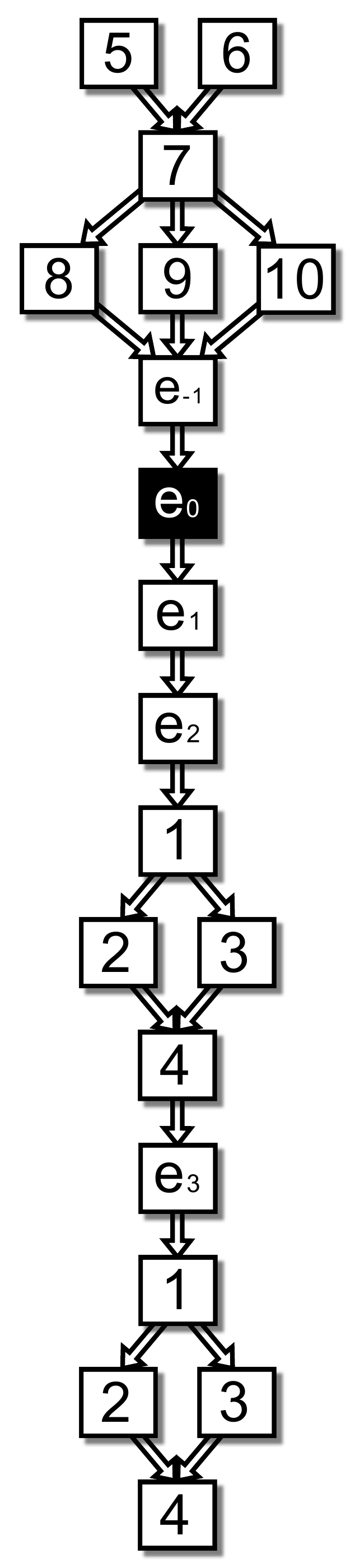}
        \caption{$\ont=\Gamma_{12}\concat{} P_{\emptyset}\concat{} P_{\emptyset}\concat{} P_{\emptyset}\concat{} P\concat{} P$}
        \label{subfig:bin_gpp}
\end{subfigure}
\begin{subfigure}[b]{0.2\textwidth}
        \centering
        \includegraphics[scale=0.35]{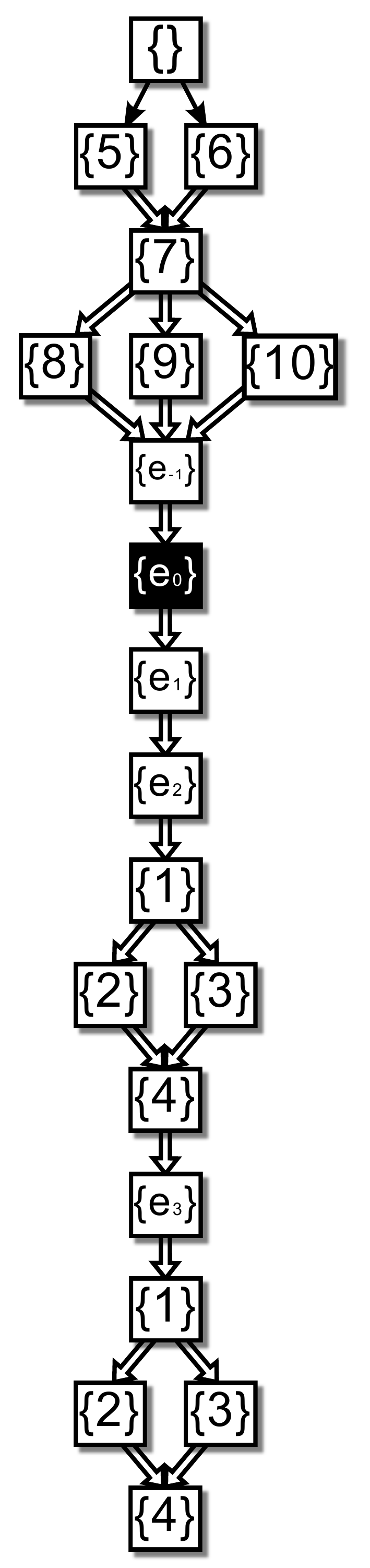}
        \caption{$\ikfunct{1}{\ont}$}
        \label{subfig:bin_igpp}
\end{subfigure}
\begin{subfigure}[b]{0.32\textwidth}
        \centering
        \includegraphics[scale=0.35]{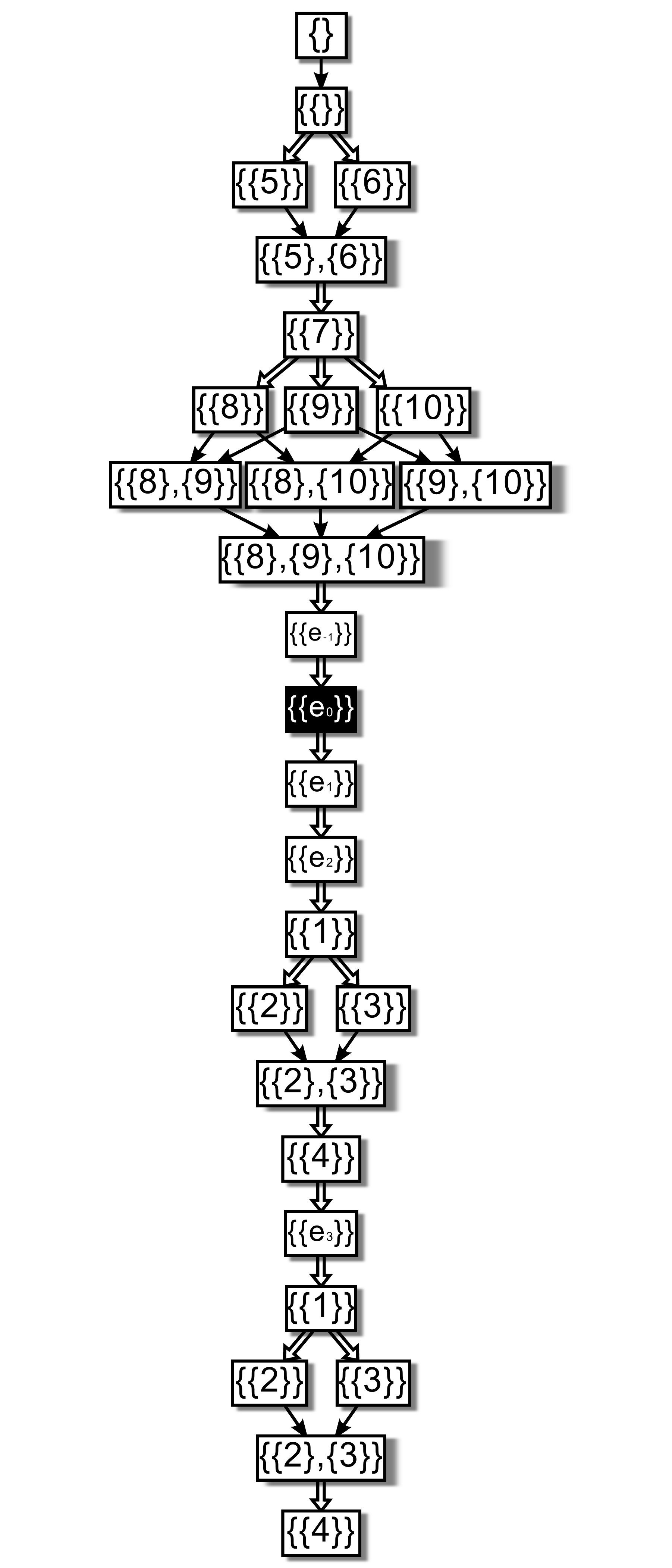}
        \caption{$\funct{S^{\left(1\right)}}{\ont}$}
        \label{subfig:bin_sgpp}
\end{subfigure}
\caption{Example posets for the proof of Lemma~\ref{lem:compare}}
\label{fig:ex_lemcomp}
\end{figure*}

\needspace{5em}
\subsection{Finding the Best-split Element}
\label{sec:best_split}
We give here the full details of the proof for Thm.~\ref{thm:bestsplit}.
We start by a few auxiliary results about the relationship between antichain counting and best-split identification, which will be needed to prove hardness.

Define the \emph{concatenation operator} $\concat{}$ on two posets $P$ and $Q$ as follows: $P\concat{} Q$ is a poset, whose elements consist of a copy of $P$ and a copy of $Q$ that are disjoint (for simplicity, we abuse notation and call these copies $P$ and $Q$), plus a new element $e$ that is neither in $P$ nor in $Q$. The order relation over $P\concat{} Q$ is such that its restriction to $P \times P$ and $Q \times Q$ matches the original order on $P$ and $Q$, and such that $p\leq e$ for every element $p$ in $P$ and $e \leq q$ for every element $q$ in $Q$. (Note that this implies that the order is total on $P \times Q$: for all $p \in P$ and $q \in Q$, we have $p \leq q$ by transitivity). Equivalently, $\concat{}$ can be defined as a \emph{series composition} of $P$, $e$ and $Q$, which is a standard operator in order theory.\footnote{See, e.g., M\"{o}hring, R.\ H. \emph{Computationally tractable classes of ordered sets}. Springer Netherlands, 1988.}
The $\concat{}$ operation is clearly associative.

We first define a few useful Lemmas, and then prove the main claim.

\begin{lemma}
\label{lem:concat}
Given two posets $P$ and $Q$ with distinct elements, $\card{\prop{P\concat{} Q}}=\card{\prop{P}}+\card{\prop{Q}}$.
\end{lemma}
\begin{proof}
Consider the antichains in $P\concat{} Q$. First, it is easy to see that $\prop{P}\cup\prop{Q}\cup\{\{e\}\}\subseteq\prop{P\concat{} Q}$. Now, every antichain that contains an element of $P$ must contain only elements from $P$, since elements of $Q \cup \{e\}$ are comparable to every element in $P$; the same applies to antichains containing an element of $Q$. As $e$ is comparable to any other element, the only antichain containing $e$ is the singleton $\{e\}$. Thus $\prop{P}\cup\prop{Q}\cup\{\{e\}\}=\prop{P\concat{} Q}$. The union on the left hand size is a disjoint union except for the empty antichain that is common to $\prop{P}$ and $\prop{Q}$ is the empty itemset; thus $\card{\prop{P}\cup\prop{Q}\cup\{\{e\}\}}=\card{\prop{P}}+\card{\prop{Q}}-1+\card{\{\{e\}\}}$.
\end{proof}

\begin{lemma}
There exists a family $(\Gamma_n)$ of posets such that for every natural number $n$,  $\card{\Gamma_n}=\oof{\log^2 n}$ and $\card{\prop{\Gamma_n}} = n$.
\end{lemma}
\begin{proof}
For $n$ s.t.\ $n=2^m$ for some natural $m$, take $\Gamma_n$ to be the flat poset with $m$ elements: the antichains of $\Gamma_n$ are exactly its subsets, so $\card{\prop{\Gamma_n}}=2^m=n$, and we have $\card{\Gamma_n}=m=\log_2 n$. Now, any natural number can be expressed in the binary form $n = n_1+n_2+\dots+n_p$ where $n_i=2^{m_i}$, $0\leq m_1 < m_2< \dots < m_p$ and $p=\oof{\log n}$. Then define $\Gamma_n= \Gamma_{n_1}\concat{} \dots \concat{} \Gamma_{n_p}$. By Lemma~\ref{lem:concat}, $\Gamma_n$ has
  $\card{\prop{\Gamma_{n_{\vphantom{p}1}}}}+\dots+\card{\prop{\Gamma_{n_p}}} = n_1+\dots+n_p=n$ antichains. The number of elements in $\Gamma_n$ is $\log n_1+ \dots+\log n_p= m_1+\dots m_p=\oof{m_p^2}=\oof{\log^2 n}$, and there are at most two edges per each $\Gamma_{n_i}$ element.
  Thus, the total size of $\Gamma_n$ is $\oof{\log^2 n}$.
\end{proof}

Using the construction of $\Gamma_n$, we can show that identifying the best split gives us some information about the number of antichains in a given poset.

\begin{lemma}
\label{lem:compare}
Assume there exists an algorithm $\mathcal{A}$
to identify the best-split element in $\ontik$ in $\oofi{\funct{poly}{\card{\ont}}}$. Then, for any given poset $P$ and $n\leq2^{\card{P}}$, $\mathcal{A}$ can decide in $\oofi{\funct{poly}{\card{P}}}$ whether $\card{\prop{P}}$ is less than, greater than or equal to $n$.
\end{lemma}
\begin{proof}
We prove the lemma by constructing a poset with two parts: one corresponding to $\Gamma_{2n}$, and one corresponding to $P \concat{} P$. Intuitively, finding the best split element allows us to compare the number of antichains in $\Gamma_{2n}$ (namely, $2n$) and the number of antichains of $P \concat{} P$, thus comparing $n$ and $\card{\prop{P}}$.
We use $2n$ and $P \concat{} P$ instead of $n$ and $P$ to ensure that the number of antichains is even, so there is only one best-split element, which simplifies the analysis.

Formally, define $\ont$ to be $\Gamma_{2n}\concat{} P_{\emptyset}\concat{} P_{\emptyset}\concat{} P_{\emptyset}\concat{} P\concat{} P$, where $P_{\emptyset}$ is an empty poset.
Let $e_{-1}$, $e_0$, $e_1$, $e_2$ and $e_3$ be the ``$e$'' elements created by the successive concatenations.
Since the size of $\Gamma_{2n}$ is $\oofi{\log^2(2n)}$ which is itself $\oofi{\card{P}^2}$,
$\ont$ can be computed in polynomial time.

Recall that the only difference between the structures of $\ont$ and $\ontik[1]$ is the additional root element, representing the empty itemset.
By Lemma~\ref{lem:concat}, the total number of antichains of $\ont$ is $2n + 3 + 2\card{\prop{P}}$, and the number of antichains of $\ontik[1]$ is $1+2n + 3 + 2\card{\prop{P}}$ (due to the addition of the root). Out of them, exactly $2\cdot\card{\prop{P}}+2$ are supersets of $\{e_0\}$: $\{\{e_0\}\}$, $\{\{e_1\}\}$, $\{\{e_2\}\}$ and the antichains of $P\concat{} P$, excluding the empty antichain.
There are exactly $2n+2$ other antichains -- the empty antichain, the antichain of the empty itemset $\{\emptyset\}$, the antichains of itemsets from $\Gamma_{2n}$, and $\{\{e_{-1}\}\}$.

Now, apply $\mathcal{A}$ on $\ont$ (and $k=1$) to obtain the best-split element $A$ in $\ontik[1]$. If $\card{\prop{P}}=n$, then the best-split element is $\{e_0\}$, because the number of antichains containing its descendants ($2\card{\prop{P}}+2$) is exactly the number of the rest of the antichains ($2n+2$), and there are no other best-split elements, since every ancestor of $\{e_0\}$ has at least one more antichain containing it or its descendants ($\{\{e_{-1}\}\}$), every descendant of $\{e_0\}$ has one less antichain containing it or its descendants ($\{\{e_0\}\}$), and every element is comparable to $e_0$ so it is either an ancestor or a descendant.
This is where we use the fact that the number of antichains above and below $\{\{e_0\}\}$ is even, since otherwise we might have gotten more than one best-split element.

In a similar manner, we can show that when $\card{\prop{P}}=n+1$, $\{e_1\}$ is the only best-split element, and when $\card{\prop{P}}=n-1$, $\{e_{-1}\}$ is the only best-split element.
In general, when $\card{\prop{P}}$ is larger than $n$ (resp., smaller than $n$), the best-split will be an descendant of $\{e_1\}$ (resp., an ancestor of $\{e_{-1}\}$).
We thus decide as follows: if $A=\{e_0\}$, $\card{\prop{P}}=n$; if $A=\{e_1\}$ or one of its descendants, $\card{\prop{P}}>n$, and otherwise $\card{\prop{P}}<n$.
\end{proof}

In order to visualize the structures used in the proof, consider the following example.
\begin{example}
Assume that we are interested in finding the number of antichains in the poset $P$ depicted in Figure~\ref{subfig:bin_p}. Since $P$ has~$4$ elements,
the number of antichains is at most $2^4=16$. Assume that we are currently trying to compare it to $n=6$. Based on the proof above, we define $\ont=\Gamma_{12}\concat{} P_{\emptyset}\concat{} P_{\emptyset}\concat{} P_{\emptyset}\concat{} P\concat{} P$. The resulting poset is depicted in Figure~\ref{subfig:bin_gpp}. $\ontik[1]$, illustrated in Figure~\ref{subfig:bin_igpp} is similar to $\ont$, but has an additional root. Now, consider $\ontsk[1]$, illustrated in Figure~\ref{subfig:bin_sgpp}. In this small example, we can count the number of descendants and non-descendants of $\{\{e_0\}\}$. Both turn out to be~$14=2n+2$. In this case the best-split element would be $\{e_0\}$, and we can determine that the number of antichains in $P$ is $6$. This is true: the antichains are $\emptyset$, $\{1\}$ $\{2\}$, $\{3\}$, $\{2,3\}$ and $\{4\}$.
\end{example}

Finally, we prove Thm.~\ref{thm:bestsplit}, namely that the problem of identifying the best-split element in $\ontik$ is $\fpsp$-complete w.r.t.\ $\card{\ont}$.

\begin{proof}[(Thm.~\ref{thm:bestsplit})]
\emph{Membership.}
Given a taxonomy $\ont$ and a number $k$, we can construct $\ontik$ in polynomial time using Algorithm~\ref{alg:construct}. Assume that we have an oracle that given a poset returns the number of antichains in it. This problem is known to be $\mathrm{\#P}$-complete~\cite{provan1983complexity}. First, use the oracle to count the
number of antichains in $\ontik$, which is equal to the number of possible solutions (sets of MFIs) for this instance of \aprob{}. Then, for each element $A$ in $\ontik$, generate a copy $\Phi$ of $\ontik$ which excludes $A$ and its descendants. The construction of $\Phi$ is naturally in PTIME. Use the oracle to count the antichains in the resulting poset. These are exactly the solutions that are possible if $A$ is infrequent. $\min\{\prop{\Phi},\prop{\ontik}-\prop{\Phi}\}$ gives the number of solutions that are guaranteed to be eliminated if $A$ is queried. The element $A$ that maximizes this formula is the best-split element, and for a constant $k$, there is a polynomial number of elements in
$\ontik$, so the best-split element is identified after a polynomial number of invocations of the oracle.

\emph{Hardness.}
Because antichain counting is an $\fpsp$-complete problem, to show $\fpsp$-hardness of finding the best split it suffices to show a PTIME reduction from best-split computation to antichain counting.
Let $\mathcal{A}$ be an oracle taking as input $\ont$ and returning the best-split element in $\ontik$. Given a poset $P$, we show how to determine its number of antichains in PTIME using oracle $\mathcal{A}$.
We perform a binary search for the number of antichains of $P$. Our initial range of values for $\card{\prop{P}}$ is $1\dots2^{\card{P}}$ because $\emptyset$ is always an antichain and the number of antichains in clearly bounded by the number of subsets of $P$. We can use $\mathcal{A}$ and the method in Lemma~\ref{lem:compare} to search within this range. In the worst case, the binary search tests ${\card{P}}$ possible values for $\card{\prop{P}}$, and performs each test in polynomial time with one invocation of $\mathcal{A}$, so we have the desired PTIME reduction.
\end{proof}

\subsection{Chain Partitioning}
\label{sec:partition_app}
We next describe a simple algorithm \textsf{partition}, which computes a partition of an itemset taxonomy $\onti$ into disjoint chains.

Start by materializing $\onti$, which can be done in time $\oofi{\card{\idomop}\cdot\card{\idomo}^2}$ by Prop.~\ref{prop:contruct}.
Based on Dilworth's theorem, chain partitioning can be computed by solving a maximum matching (or maximal join) in a bipartite graph~\cite{fulkerson1956note}. The construction involves creating
one edge for each pair of itemsets $A\orelit B$; this requires computing the full (transitive) $\orelit$ relation in time $\oofi{\card{\idomop}^2}$.
Finding the matching can be done, e.g., in
$\oofi{\card{\onti}\sqrt{\card{\idomop}}}$, using the Hopcroft-Karp algorithm\footnote{Hopcroft, John E., and Richard M. Karp. (1973). ``An $n^{5/2}$ algorithm for maximum matchings in bipartite graphs.'' \textit{SIAM Journal on computing} 2(4): 225--231.}
By summing the complexity of all the steps and simplifying the result, we obtain the algorithm complexity $\oofi{\card{\idomop}\cdot\card{\idomo}^2+\card{\idomop}^2+\card{\onti}\sqrt{\card{\idomop}}}$. It may be possible to further simplify this expression, depending on the structure of $\ont$.

\subsection{Greedy Algorithms}
\label{sec:greedy_app}
\begin{proof} [(Thm.~\ref{prop:greedy})]
We want to prove that finding the greedy best-split itemset is $\fpsp$-hard w.r.t.\ $\card{\ont}$.

\emph{Observation~1}. The structure of $\ontsk[1]=\ifunct{\ontik[1]}$ is almost identical to that of $\onti$, up to an additional root element. This follows from the fact that the structures of $\ont$ and $\ontik[1]$ are almost identical.

\emph{Observation~2}. If the itemset $A$ is frequent, then all the solutions that do not contain any $B\in\funct{desc}{A}$ are impossible. Similarly, if $A$ is infrequent, all the solutions that contain some $B\in\funct{desc}{A}$ are impossible. Consequently, querying $A$ eliminates half of the solution space iff $\{A\}$ has the same number of descendants and non-descendants
in the solution taxonomy.

Using these observations, we show a PTIME reduction from finding the greedy best-split element to counting antichains. Let $\mathcal{A}$ be an oracle taking a taxonomy $\ont$ as input and returning the greedy best-split of $\onti$.
Given a poset $P$, we show how to determine its number of antichains in PTIME using oracle $\mathcal{A}$.
As in the proof of Thm.~\ref{thm:bestsplit}, define $\ont=\Gamma_{2n}\concat{} P_{\emptyset}\concat{} P_{\emptyset}\concat{} P_{\emptyset}\concat{} P\concat{} P$, where $P$ is the input to the antichain counting problem.
Define $\ont' = \ontik[1]$.
It is enough to show that by finding the greedy best-split element in $\ifunct{\ont'}=\ontsk[1]$, we can decide whether
the best-split element in $\ontik[1]$ is $\{e_0\}$, one of its ancestors or one of its descendants. We have three cases:
\begin{compactenum}
  \item The number of descendants and non-descendants of $\{\{e_0\}\}$ in $\ontsk[1]$ is identical. This allows us to deduce two things. First, by observation~2, $\{e_0\}$ is the best-split element of $\ont' = \ontik[1]$. Second, $\{\{e_0\}\}$ is the greedy best-split element of $\ifunct{\ont'}$, because the non-descendants of $\{\{e_0\}\}$ are exactly its ancestors since $\{\{e_0\}\}$ is comparable to all other elements of $\ontsk[1]$.
  \item $\{\{e_0\}\}$ has more descendants than non-descendants in $\ifunct{\ont'}$. In this case, by a similar reasoning we can show that the greedy best-split element of $\ifunct{\ont'}$ must be $\{\{e_1\}\}$ or one of its descendants, and using observation~2 we can show that the best-split element of $\ontik[1]$ is one of $\{e_1\}$'s descendants.
  \item $\{\{e_0\}\}$ has less descendants than non-descendants. Similarly to the previous case, we can show that the greedy best-split element of $\ifunct{\ont'}$ must be $\{\{e_{-1}\}\}$ or one of its ancestors and that the best-split element of $\ontik[1]$ is $\{e_{-1}\}$ or one of its ancestors.
\end{compactenum}
So, by applying $\mathcal{A}$ to $\ont'$ we can determine if $\{e_0\}$ in $\ontik[1]$ is the best-split element. We can therefore
count the number of antichains in $P$
in a similar manner to the proof of Thm.~\ref{thm:bestsplit}.
\end{proof}